\documentclass{article}

\usepackage{arxiv}
\usepackage{amsthm}
\usepackage[utf8]{inputenc} 
\usepackage{float}
\usepackage{latexsym}
\usepackage{xspace}
\usepackage{physics}
\usepackage{comment}
\usepackage{diagbox}
\usepackage{makecell}
\usepackage[ruled,lined]{algorithm2e}
\usepackage{tikz}
\usetikzlibrary{automata}
\usetikzlibrary{cd}
\usetikzlibrary{quantikz}
\usetikzlibrary{arrows,arrows.meta}
\usetikzlibrary{graphs,graphs.standard}
\usetikzlibrary{shapes}
\tikzcdset{arrow style=tikz, diagrams={>= latex}}
\usepackage{adjustbox}
\usepackage[justification=centering,labelfont=rm,caption=false]{subfig}
\usepackage{pgfplots}
\usepackage{enumitem}
\usepackage{multirow}
\usepackage{xcolor}
\usepackage{graphicx}

\newtheorem{Theorem}{Theorem}[section]
\newtheorem{Corollary}[Theorem]{Corollary}
\newtheorem{Lemma}[Theorem]{Lemma}
\newtheorem{Proposition}[Theorem]{Proposition}
\newtheorem{Definition}[Theorem]{Definition}

\title{The Complexity of Quantum Circuit Mapping with Fixed Parameters}

\author{{Pengcheng Zhu}\\
	School of Information Science and Technology\\
	Nantong University\\
	Nantong, 226000, China\\
	\texttt{zhupc@ ntu.edu.cn}\\
	\And
	{Shenggen Zheng} \\
	Peng Cheng Laboratory\\
	Shenzhen 518055, China\\
	\texttt{zhengshg@pcl.ac.cn}\\
	\And
	{Lihua Wei} \\
	Department of Information and Computing Science\\
	Suqian University \\
	Suqian, 223800, China\\
	\texttt{wlh\_sqc@163.com}\\
	\And
	{Xueyun Cheng} \\
	School of Information Science and Technology\\
	Nantong University\\
	Nantong, 226000, China\\
	\texttt{chen.xy@ntu.edu.cn}\\
	\And
	{Zhijin Guan} \\
	School of Information Science and Technology\\
	Nantong University\\
	Nantong, 226000, China\\
	\texttt{guan.zj@ntu.edu.cn}\\
	\And
	{Shiguang Feng} \\
	School of Information Science and Technology\\
	Nantong University\\
	Nantong, 226000, China\\
	\texttt{shigfeng@ntu.edu.cn}\\
}

\renewcommand{\shorttitle}{\textit{arXiv} Template}

\begin{document}
	\maketitle
	
	\begin{abstract}
		A quantum circuit must be preprocessed before implementing on NISQ devices due to the connectivity constraint. Quantum circuit mapping (QCM) transforms the circuit into an equivalent one that is compliant with the NISQ device's architecture constraint by adding SWAP gates. The QCM problem asks the minimal number of auxiliary SWAP gates, and is NP-complete. The complexity of QCM with fixed parameters is studied in the paper. We give an exact algorithm for QCM, and show that the algorithm runs in polynomial time if the NISQ device's architecture is fixed. If the number of qubits of the quantum circuit is fixed, we show that the QCM problem is NL-complete by a reduction from the undirected shortest path problem.
		Moreover, the fixed-parameter complexity of QCM is W[1]-hard when parameterized by the number of qubits of the quantum circuit. We prove the result by a reduction from the clique problem. If taking the depth of the quantum circuits and the coupling graphs as parameters, we show that the QCM problem is still NP-complete over shallow quantum circuits, and planar, bipartite and degree bounded coupling graphs.
	\end{abstract}

	\keywords{qubit mapping \and quantum circuit transformation \and parameterized complexity \and NISQ device}

\section{Introduction}

Quantum computing is a new type of computation by employing the laws of quantum mechanics to perform complex tasks. It shows theoretical advantages over classical computing on some problems such as integer factorization and unsorted database search. Quantum computing has applications in many fields including cryptography, chemistry, biology, artificial intelligence and others.
The number of qubits in quantum computers has increased steadily in the past several years. Now is the noisy intermediate-scale quantum (NISQ) era~\cite{Preskill2018}. The quantum devices have dozens to hundreds of qubits with rather limited coherence time, and only support a few kinds of elementary quantum gates with nonnegligible errors. Moreover, current NISQ devices have the connectivity constraints which require that any two-qubit operation can only be applied to adjacent qubits. 
A quantum circuit must be transformed into a functionally equivalent circuit that is compliant with the connectivity constraint of the quantum device before executing~\cite{maslov2007quamplace}. 
Quantum circuit mapping (QCM) is the process that transforms the circuit by adding SWAP gates~\cite{childs2019circuit}. Due to the error caused by the decoherence and inherent noise in the NISQ device, the number of auxiliary SWAP gates should be as small as possible.
The QCM problem is, given a quantum circuit, a quantum architecture (denoted by a coupling graph in the paper), and a number $k$, asking whether $k$ SWAP gates is enough to transform the quantum circuit into an equivalent one that is compliant with the quantum architecture's constraint.
Deciding the minimal number of swap actions is an NP-complete problem~\cite{Botea2018, Marcos2018}, which implies that it is unlikely to find a polynomial time algorithm. There are many algorithms for QCM based on heuristic and approximation methods~\cite{li2019tackling,wille2019mapping,Zulehner2017exact,zhou2020monte,Chhangte2022mapibm}, which provide efficient solutions by taking advantage of the inner structural features of the quantum circuits and quantum devices. For example, mapping the logical qubits interacting with other qubits frequently to the physical qubits with high connecting degree can reduce the number of SWAP gates dramatically in the transformation~\cite{zulehner2019,niu2020hardware,zhu2021iterated}, and by exploring a fixed number of layers in the circuit in advance can give a local optimal solution without the large time consumption~\cite{li2020qubit,zhu2020dynamic}. So study the effects of the topology of quantum circuits and quantum architectures on the complexity of QCM is of importance for practical applications.

The paper~\cite{Beals2013efficient} studied the time-overhead caused by adding gates to permute qubits on 1D, 2D and hypercube quantum architectures, and~\cite{Brierley2017efficient} proved the time-overhead over cyclic butterfly architectures, whose coupling graphs are 4-degree bounded. We study systematically the complexity of QCM with the number of qubits, the depth of the quantum circuits, the topology of quantum architectures and the number of SWAP gates needed as the parameters, respectively, in this paper. First an exact algorithm that computes the minimal number of SWAP gates for the transformation is given. We show that the algorithm runs in polynomial time if either the number of qubits in the quantum circuit or the number of nodes in the coupling graph is bounded by a constant. Moreover, QCM is NL-complete if the the number of qubits in the quantum circuit is fixed. The result is obtained by a reduction from the undirected shortest path problem for graphs with maximum degree 3, which proved to be NL-complete in the paper. We also show that QCM is W[1]-hard parameterized by the number of qubits. The result is obtained by a reduction from the clique problem. The $W[i]\,(i\geq 0)$ is a set of parameterized problems to capture the fixed-parameter intractability. Every $W[j]\,(j\geq 1)$-hard problem is believed to be fixed-parameter intractable. Hence, the QCM problem is fixed-parameter intractable, and unlikely to have a fully polynomial-time approximation scheme.

The depth of the quantum circuits and the coupling graphs are the other two parameters considered in the paper. The motivation is that shallow quantum circuits are the kind of quantum circuits with fixed depth. They are powerful for some work and easily to implement on NISQ devices due to the short running time~\cite{bravyi2020}. On the other hand, most of the current NISQ device's architectures have low connecting degree, e.g., the IBM QX20 Tokyo architecture, the Rigetti 16Q-Aspen architecture, and the IBM heavy-hex lattice architecture.
We show that the QCM problem is still NP-complete on shallow quantum circuits. The result also holds if the input coupling graphs are either planar, bipartite and degree bounded, or grid graphs that are finite induced subgraphs of the (infinite) grids made up of the squares, regular hexagons and equilateral triangles, respectively. These results are obtained by reductions from two famous NP-complete problems: the Hamiltonian path problem and the Hamiltonian cycle problem. We also prove that QCM is NP-complete if the number of SWAP gates allowed is fixed. The main results in the paper are summarized in Table~\ref{tab-complexy}. Combing these results we see that the number of qubits is the key factors that affect the complexity of QCM. Finding efficient algorithms for QCM on quantum circuits with a reasonable number of qubits is theoretically possible.

\begin{table}[H]
	\centering
	\caption{The complexity of QCM with fixed parameters. \label{tab-complexy}}
	\footnotesize
	\begin{tabular}{|c|c|c|c|c|c|}
		\hline
		\diagbox{\makecell[c]{Quantum\\ circuits}}{\makecell[c]{Coupling\\ graphs}}& all & \makecell[c]{4-degree\\ bounded} & \makecell[c]{planar, \\bipartite,\\ 3-degree\\ bounded} & \makecellbox[c]{grid \\ graphs}& \makecell[c]{fixed\\ number\\ of nodes} \\ \hline
		fixed number of qubits &  NL-complete    &    NL    &  NL           & NL &  NL  \\ \hline
		depth of 3       &  NP-complete   &   NP-complete   &  NP-complete  & NP-complete &  P  \\ \hline
		depth of 2       &  NP-complete   &   NP-complete   &   - & - &    P   \\ \hline
	\end{tabular}
\end{table}

The paper is organized as follows. In Section~\ref{sec-pre}, we give the basic definitions about quantum circuits and quantum circuit mapping. In Section~\ref{sec-complexity}, we study the complexity of QCM with fixed parameters. We give an exact algorithm for QCM to compute the minimal number of SWAP gates in Subsection~\ref{subsec-exactalgo}, and consider its complexity when the number of qubits is bounded in Subsection~\ref{subsec-qubitnum}. In Subsection~\ref{subsec-paracomplxy}, we prove the fixed-parameter complexity of QCM. In Subsection~\ref{subsec-depthcoupl}, we consider the complexity of QCM with the depth of quantum circuits and the coupling graphs as parameters. In Subsection~\ref{subsec-numberofswap}, we prove the complexity of QCM when the number of SWAP gates is fixed. Finally, we conclude the paper in Section~\ref{sec-discus}.

\section{Preliminaries}\label{sec-pre}
\subsection{Basic definitions and notations}
A \textit{qubit} is the basic unit of quantum computing. The classical bit represents a logical state which has value either ``$0$'' or ``$1$'',  a qubit can be in a state $\ket{\phi} = \alpha \ket{0} + \beta \ket{1}$ that is the superposition of two basis states $\ket{0}$ and $\ket{1}$, where $\alpha,\beta$ are complex numbers and $|\alpha|^2 + |\beta|^2 =1$. A \textit{quantum gate} is a unitary operator on a group of qubits. A \textit{quantum circuit} is composed of a set of qubits and a sequence of quantum gates that operate on these qubits. The quantum circuit is a common model for quantum computation, and can be described as a gate-array diagram, in which qubits are represented as horizontal lines and quantum gates are different blocks that operate on those lines. 
Figure~\ref{fig-cirdep-a} is a quantum circuit that contains 4 qubits and 7 quantum gates.
A set of quantum gates are \textit{universal} if any quantum gate can be decomposed to a combination of quantum gates from it. 
The two-qubit CNOT gate and all single-qubit gates are widely used universal quantum gates.
In this paper, we focus only on the quantum circuits composed of two-qubit quantum gates and single-qubit quantum gates.
We use the convention that the leftmost gate in the quantum circuit executes first, and denote a quantum circuit by $\langle Q, \Gamma \rangle$, where $Q$ is the set of qubits and $\Gamma$ is a sequence of gates sorted by the execution order.

\begin{Definition}
	The topology graph of a quantum circuit $\langle Q, \Gamma \rangle$ is an undirected graph $(Q,E_t)$, where $E_t$ is a set of edges over the qubits in $Q$, such that for any two qubits $q_i$ and $q_j$, $(q_i,q_j)\in E_t$ if and only if $q_i$ and $q_j$ are operated by a two-qubit gate in $\Gamma$. 
\end{Definition}

\begin{Definition}
	 The dependency graph of a quantum circuit  $\langle Q, \Gamma \rangle$ is an acyclic directed graph $(\Gamma,E_d)$, where every gate in $\Gamma$ is a node, and $E_d$ is a set of edges over the gates in $\Gamma$ such that for any two gates $g_i$ and $g_j$, $(g_i,g_j)\in E_d$ if and only if an output of $g_i$ is an input of $g_j$. 
\end{Definition}

\begin{figure}[h]
	\centering
	\begin{adjustbox}{width=0.35\textwidth}
		\begin{quantikz}
			\lstick{$q_0$}  &  \gate{H}\gategroup[wires=1,steps=1,style={dashed, rounded corners, inner sep=1pt}]{$g_1$}  
			&  [0.5cm] \ctrl{1}\gategroup[wires=2,steps=1,style={dashed, rounded corners, inner xsep=5pt, inner ysep=1pt}]{$g_3$}  
			&  [0.5cm]\qw  
			&  [0.5cm]\ctrl{1}\gategroup[wires=2,steps=1,style={dashed, rounded corners, inner sep=1pt}]{$g_6$}  
			&  [0.5cm]\qw \\		
			\lstick{$q_1$}  &  \qw     
			&  [0.5cm]\targ{}  
			&  [0.5cm]\targ{}\gategroup[wires=3,steps=1,style={dashed, rounded corners, inner xsep=5pt, inner ysep=1pt}]{$g_5$} 
			&  [0.5cm]\targ{} 
			&  [0.5cm]\qw \\
			\lstick{$q_2$}  &  \gate{T}\gategroup[wires=1,steps=1,style={dashed, rounded corners, inner sep=1pt}]{$g_2$}
			&  [0.5cm]\ctrl{1}\gategroup[wires=2,steps=1,style={dashed, rounded corners, inner xsep=5pt, inner ysep=1pt}, label style={label position=below, yshift=-0.4cm}]{$g_4$} 
			&  [0.5cm]\qw       
			&  [0.5cm]\qw 
			&  [0.5cm]\qw \\
			\lstick{$q_3$}  &  \qw      
			&  [0.5cm]\targ{}   
			&  [0.5cm]\ctrl{-2}
			&  [0.5cm]\gate{H}\gategroup[wires=1,steps=1,style={dashed, rounded corners, inner sep=1pt}, label style={label position=below, yshift=-0.4cm}]{$g_7$} 
			&  [0.5cm]\qw
		\end{quantikz}
	\end{adjustbox}
	\hspace{1cm}
	\begin{adjustbox}{width=0.16\textwidth}
		\begin{tikzcd}[cells={nodes={draw, circle}}]
			g_1 \arrow[r] & g_3 \arrow[r] \arrow[rd] & g_6  \\
			g_2 \arrow[r] & g_4 \arrow[r] & g_5 \arrow[d] \arrow[u] \\
			&						  & g_7
		\end{tikzcd}
	\end{adjustbox}
	\subfloat[\label{fig-cirdep-a}]{\hspace{0.4\linewidth}}
	\subfloat[\label{fig-cirdep-b}]{\hspace{0.3\linewidth}}
	\bigskip
	
	\begin{adjustbox}{width=0.35\textwidth}
		\begin{quantikz}
			\lstick{$q_0$}  &  \gate{H}\gategroup[wires=4,steps=1,style={dashed, inner sep=1pt}, label style={yshift=0.2cm}]{layer1}  
			&  [0.5cm] \ctrl{1}\gategroup[wires=4,steps=1,style={dashed, inner xsep=5pt, inner ysep=1pt}, label style={yshift=0.2cm}]{layer2}  
			&  [0.5cm] \qw  \gategroup[wires=4,steps=1,style={dashed, inner xsep=5pt, inner ysep=1pt}, label style={yshift=0.2cm}]{layer3}  
			&  [0.5cm]\ctrl{1}\gategroup[wires=4,steps=1,style={dashed, inner sep=1pt}, label style={yshift=0.2cm}]{layer4}  
			&  [0.5cm]\qw \\		
			\lstick{$q_1$}  &  \qw     
			&  [0.5cm]\targ{}  
			&  [0.5cm]\targ{}
			&  [0.5cm]\targ{} 
			&  [0.5cm]\qw \\
			\lstick{$q_2$}  &  \gate{T}
			&  [0.5cm]\ctrl{1}
			&  [0.5cm]\qw       
			&  [0.5cm]\qw 
			&  [0.5cm]\qw \\
			\lstick{$q_3$}  &  \qw      
			&  [0.5cm]\targ{}   
			&  [0.5cm]\ctrl{-2}
			&  [0.5cm]\gate{H}
			&  [0.5cm]\qw
		\end{quantikz}
	\end{adjustbox}
	\hspace{1cm}
	\begin{adjustbox}{width=0.16\textwidth}
		\hspace{0.1cm}
		\begin{tikzcd}[cells={nodes={draw, circle}},every arrow/.append style={dash}]
			q_0 \arrow[r] & q_1 \arrow[d]   \\
			q_2 \arrow[r] & q_3 
		\end{tikzcd}
		\hspace{0.5cm}
	\end{adjustbox}
	\subfloat[\label{fig-cirdep-c}]{\hspace{0.4\linewidth}}
	\subfloat[\label{fig-cirdep-d}]{\hspace{0.3\linewidth}}
	\caption{\label{fig-cirdep} The illustration of a quantum circuit (\textbf{a}), and (\textbf{b}) its dependency graph, (\textbf{c}) its layer partition, (\textbf{d}) is topology graph.}
\end{figure}
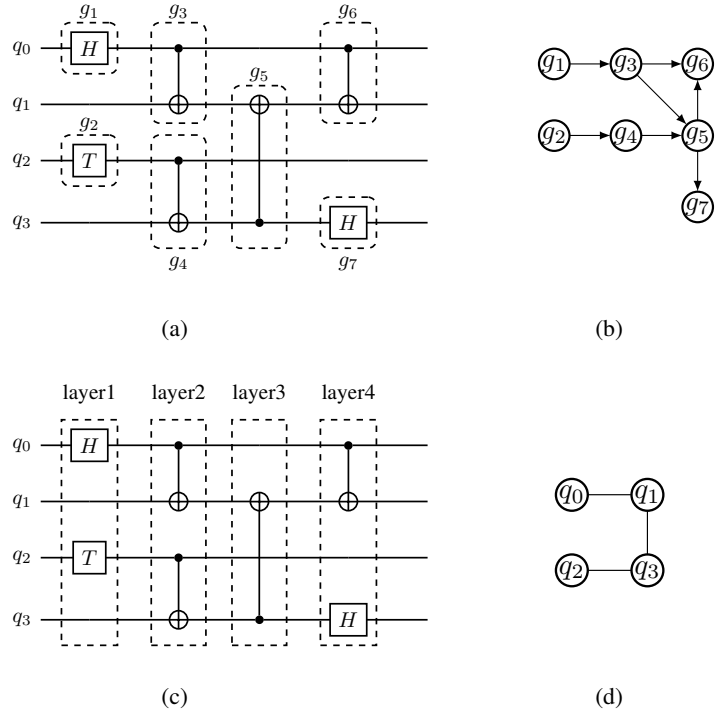

Figure~\ref{fig-cirdep-b} and Figure~\ref{fig-cirdep-d} are the dependency graph and the topology graph of Figure~\ref{fig-cirdep-a}, respectively. Let $\Phi$ be a quantum circuit, and $g_i,g_j$ two gates in $\Phi$. We say that $g_j$ \textit{depends on} $g_i$ if there is a path from $g_i$ to $g_j$ in the dependency graph of $\Phi$.

\begin{Definition}
	A quantum circuit $\Psi$ is a subcircuit of the quantum circuit $\Phi$ if $\Psi$ is obtained by removing some gates in $\Phi$ and all gates that they depend on are also be removed. 
\end{Definition}

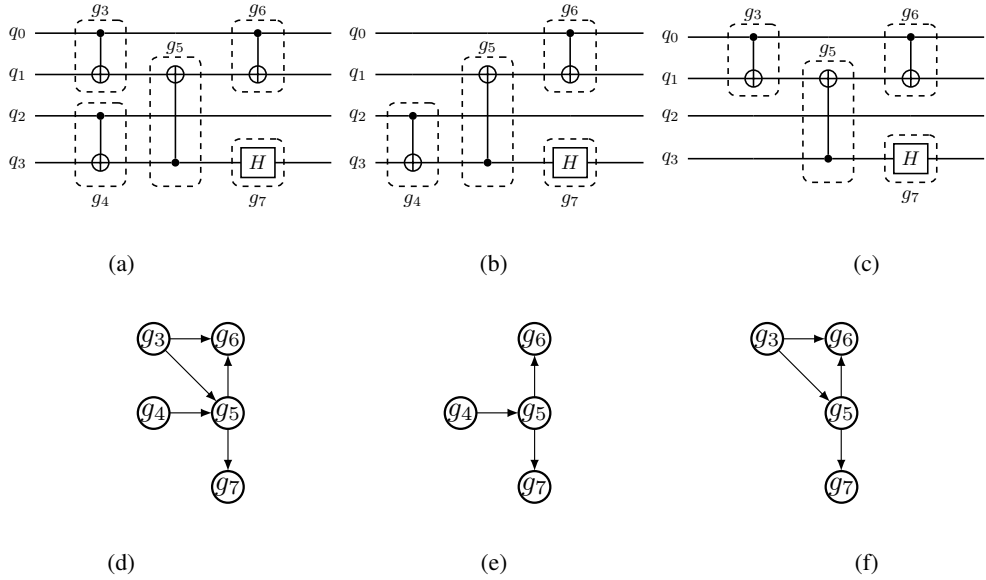
\begin{figure}[h]
	\centering
	\begin{adjustbox}{width=0.8\textwidth}
		\begin{quantikz}
			\lstick{$q_0$} &  [0.5cm] \ctrl{1}\gategroup[wires=2,steps=1,style={dashed, rounded corners, inner xsep=5pt, inner ysep=1pt}]{$g_3$}  
			&  [0.5cm]\qw  
			&  [0.5cm]\ctrl{1}\gategroup[wires=2,steps=1,style={dashed, rounded corners, inner sep=1pt}]{$g_6$}  
			&  [0.5cm]\qw \\		
			\lstick{$q_1$} &  [0.5cm]\targ{}  
			&  [0.5cm]\targ{}\gategroup[wires=3,steps=1,style={dashed, rounded corners, inner xsep=5pt, inner ysep=1pt}]{$g_5$} 
			&  [0.5cm]\targ{} 
			&  [0.5cm]\qw \\
			\lstick{$q_2$} &  [0.5cm]\ctrl{1}\gategroup[wires=2,steps=1,style={dashed, rounded corners, inner xsep=5pt, inner ysep=1pt}, label style={label position=below, yshift=-0.4cm}]{$g_4$} 
			&  [0.5cm]\qw       
			&  [0.5cm]\qw 
			&  [0.5cm]\qw \\
			\lstick{$q_3$} &  [0.5cm]\targ{}   
			&  [0.5cm]\ctrl{-2}
			&  [0.5cm]\gate{H}\gategroup[wires=1,steps=1,style={dashed, rounded corners, inner sep=1pt}, label style={label position=below, yshift=-0.4cm}]{$g_7$} 
			&  [0.5cm]\qw
		\end{quantikz}
		\begin{quantikz}
			\lstick{$q_0$} & \qw
			&  [0.5cm]\qw  
			&  [0.5cm]\ctrl{1}\gategroup[wires=2,steps=1,style={dashed, rounded corners, inner sep=1pt}]{$g_6$}  
			&  [0.5cm]\qw \\		
			\lstick{$q_1$} &  [0.5cm]\qw
			&  [0.5cm]\targ{}\gategroup[wires=3,steps=1,style={dashed, rounded corners, inner xsep=5pt, inner ysep=1pt}]{$g_5$} 
			&  [0.5cm]\targ{} 
			&  [0.5cm]\qw \\
			\lstick{$q_2$} &  [0.5cm]\ctrl{1}\gategroup[wires=2,steps=1,style={dashed, rounded corners, inner xsep=5pt, inner ysep=1pt}, label style={label position=below, yshift=-0.4cm}]{$g_4$} 
			&  [0.5cm]\qw       
			&  [0.5cm]\qw 
			&  [0.5cm]\qw \\
			\lstick{$q_3$} &  [0.5cm]\targ{}   
			&  [0.5cm]\ctrl{-2}
			&  [0.5cm]\gate{H}\gategroup[wires=1,steps=1,style={dashed, rounded corners, inner sep=1pt}, label style={label position=below, yshift=-0.4cm}]{$g_7$} 
			&  [0.5cm]\qw
		\end{quantikz}
		\begin{quantikz}
			\lstick{$q_0$} &  [0.5cm] \ctrl{1}\gategroup[wires=2,steps=1,style={dashed, rounded corners, inner xsep=5pt, inner ysep=1pt}]{$g_3$}  
			&  [0.5cm]\qw  
			&  [0.5cm]\ctrl{1}\gategroup[wires=2,steps=1,style={dashed, rounded corners, inner sep=1pt}]{$g_6$}  
			&  [0.5cm]\qw \\		
			\lstick{$q_1$} &  [0.5cm]\targ{}  
			&  [0.5cm]\targ{}\gategroup[wires=3,steps=1,style={dashed, rounded corners, inner xsep=5pt, inner ysep=1pt}]{$g_5$} 
			&  [0.5cm]\targ{} 
			&  [0.5cm]\qw \\
			\lstick{$q_2$} &  [0.5cm]\qw
			&  [0.5cm]\qw       
			&  [0.5cm]\qw 
			&  [0.5cm]\qw \\
			\lstick{$q_3$} &  [0.5cm]\qw 
			&  [0.5cm]\ctrl{-2}
			&  [0.5cm]\gate{H}\gategroup[wires=1,steps=1,style={dashed, rounded corners, inner sep=1pt}, label style={label position=below, yshift=-0.4cm}]{$g_7$} 
			&  [0.5cm]\qw
		\end{quantikz}
	\end{adjustbox}
	\subfloat[\label{fig-subcirdep-a}]{\hspace{0.3\linewidth}}
	\subfloat[\label{fig-subcirdep-b}]{\hspace{0.3\linewidth}}
	\subfloat[\label{fig-subcirdep-c}]{\hspace{0.3\linewidth}}
	\bigskip
	
	\begin{adjustbox}{width=0.6\textwidth}
		\begin{tikzcd}[cells={nodes={draw, circle}}]
			g_3 \arrow[r] \arrow[rd] & g_6  \\
			g_4 \arrow[r]			 & g_5 \arrow[d] \arrow[u] \\
			& g_7
		\end{tikzcd}
		\hspace{2cm}
		\begin{tikzcd}[cells={nodes={draw, circle}}]
			& g_6  \\
			g_4	 \arrow[r]		   & g_5 \arrow[d] \arrow[u] \\
			& g_7
		\end{tikzcd}
		\hspace{2cm}
		\begin{tikzcd}[cells={nodes={draw, circle}}]
			g_3 \arrow[r] \arrow[rd]   &  g_6  \\
			&  g_5 \arrow[d] \arrow[u] \\
			&  g_7
		\end{tikzcd}
	\end{adjustbox}
	\subfloat[\label{fig-subcirdep-d}]{\hspace{0.3\linewidth}}
	\subfloat[\label{fig-subcirdep-e}]{\hspace{0.3\linewidth}}
	\subfloat[\label{fig-subcirdep-f}]{\hspace{0.3\linewidth}}
	\bigskip
	\caption{\label{fig-subcirdep} The illustration of subcircuits. (\textbf{a}), (\textbf{b}), (\textbf{c}) are there subcircuits of Figure~\ref{fig-cirdep-a}. (\textbf{d}), (\textbf{e}), (\textbf{f}) are their dependency graphs respectively.}
\end{figure}

The quantum circuits (a), (b), (c) in Figure~\ref{fig-subcirdep} are three subcircuits of Figure~\ref{fig-cirdep-a}, and (d), (e), (f) are their dependency graphs, respectively.
Let $\Psi_1$ and $\Psi_2$ be two subcircuits of $\Phi$. 
We say that $\Psi_1$ is \textit{smaller than} $\Psi_2$ if $\Psi_1$ is a subcircuit of $\Psi_2$. $\Psi_1$ and $\Psi_2$ are \textit{incomparable} if neither of them is smaller than the other one. 
Let $S$ be a set of subcircuits of $\Phi$. 
A subcircuit $\Psi_1$ in $S$ is \textit{minimal} if no subcircuit in $S$ is smaller than $\Psi_1$.
The set $S$ is \textit{minimized} if all subcircuits in it are minimal. Every pair of subcircuits in a minimized set are incomparable,
and we can minimize a subcircuit set by removing all non-minimal subcircuits in it.

A quantum circuit can be uniquely divided into \textit{layers} which are maximal sets of gates such that all gates in the same layer can execute in parallel. The quantum circuit of Figure~\ref{fig-cirdep-a} has 4 layers as illustrated in Figure~\ref{fig-cirdep-c}.
The \textit{depth} of a quantum circuit is the number of its layers. We use $depth(\Phi)$ to denote the depth of the quantum circuit $\Phi$.

\subsection{Quantum circuit mapping}

In a quantum circuit, all qubits are assumed to be all-to-all connected. However, in many quantum computers the connectivity of qubits has been reduced, and only adjacent qubits can interact with each other. A quantum computer can be abstracted by a graph that shows the connectivity between qubits.

\begin{figure}[h]
	\centering
	\begin{adjustbox}{width=0.25\textwidth}
		\begin{tikzcd}[every arrow/.append style={dash}]
			\tikz{\node[draw, circle, inner sep=2pt]{$Q_0$}}  \arrow[r] \arrow[d] & \tikz{\node[draw, circle, inner sep=2pt]{$Q_1$}}  \arrow[r] \arrow[d]
			\arrow[rd,shorten <= -4pt, shorten >= -4pt] & \tikz{\node[draw, circle, inner sep=2pt]{$Q_2$}}  \arrow[r] \arrow[d] \arrow[ld, shorten <= -4pt, shorten >= -4pt] & \tikz{\node[draw, circle, inner sep=2pt]{$Q_3$}} \arrow[r] \arrow[d] \arrow[rd, shorten <= -4pt, shorten >= -4pt] & \tikz{\node[draw, circle, inner sep=2pt]{$Q_4$}}  \arrow[d] \arrow[ld, shorten <= -4pt, shorten >= -4pt]  \\
			\tikz{\node[draw, circle, inner sep=2pt]{$Q_5$}}  \arrow[r] \arrow[d] \arrow[rd, shorten <= -4pt, shorten >= -4pt] & \tikz{\node[draw, circle, inner sep=2pt]{$Q_6$}}  \arrow[r] \arrow[d] 
			\arrow[ld, shorten <= -4pt, shorten >= -4pt] & \tikz{\node[draw, circle, inner sep=2pt]{$Q_7$}}  \arrow[r] \arrow[d] \arrow[rd, shorten <= -4pt, shorten >= -4pt] & \tikz{\node[draw, circle, inner sep=2pt]{$Q_8$}}   \arrow[r] \arrow[d] \arrow[ld, shorten <= -4pt, shorten >= -4pt] & \tikz{\node[draw, circle, inner sep=2pt]{$Q_9$}} \arrow[d]  \\
			\tikz{\node[draw, circle, inner sep=1pt]{$Q_{10}$}} \arrow[r] \arrow[d] & \tikz{\node[draw, circle, inner sep=1pt]{$Q_{11}$}} \arrow[r] \arrow[d] \arrow[rd, shorten <= -4pt, shorten >= -4pt] & \tikz{\node[draw, circle, inner sep=1pt]{$Q_{12}$}} \arrow[r] \arrow[d] \arrow[ld, shorten <= -4pt, shorten >= -4pt] & \tikz{\node[draw, circle, inner sep=1pt]{$Q_{13}$}} \arrow[r] \arrow[d] 
			\arrow[rd, shorten <= -4pt, shorten >= -4pt] & \tikz{\node[draw, circle, inner sep=1pt]{$Q_{14}$}} \arrow[d] \arrow[ld, shorten <= -4pt, shorten >= -4pt] \\
			\tikz{\node[draw, circle, inner sep=1pt]{$Q_{15}$}} \arrow[r] & \tikz{\node[draw, circle, inner sep=1pt]{$Q_{16}$}} \arrow[r] & \tikz{\node[draw, circle, inner sep=1pt]{$Q_{17}$}} \arrow[r] & \tikz{\node[draw, circle, inner sep=1pt]{$Q_{18}$}} \arrow[r] & \tikz{\node[draw, circle, inner sep=1pt]{$Q_{19}$}}  
		\end{tikzcd}
	\end{adjustbox}
	\hspace{1cm}
	\begin{adjustbox}{width=0.4\textwidth}
		\begin{tikzcd}[every arrow/.append style={dash}]
			& \tikz{\node[draw, circle, inner sep=2pt]{$Q_1$}} \arrow[r] \arrow[ld,shorten <= -4pt, shorten >= -4pt] &\tikz{\node[draw, circle, inner sep=2pt]{$Q_2$}} \arrow[rd, shorten <= -4pt, shorten >= -4pt] &  &  & \tikz{\node[draw, circle, inner sep=2pt]{$Q_3$}} \arrow[r] \arrow[ld, shorten <= -4pt, shorten >= -4pt] & \tikz{\node[draw, circle, inner sep=2pt]{$Q_4$}} 
			\arrow[rd, shorten <= -4pt, shorten >= -4pt] & \\ 
			\tikz{\node[draw, circle, inner sep=2pt]{$Q_5$}} \arrow[d] &  &  & \tikz{\node[draw, circle, inner sep=2pt]{$Q_6$}} \arrow[r] \arrow[d] & \tikz{\node[draw, circle, inner sep=2pt]{$Q_7$}} \arrow[d]   &     &    & \tikz{\node[draw, circle, inner sep=2pt]{$Q_8$}} \arrow[d]  \\
			\tikz{\node[draw, circle, inner sep=2pt]{$Q_9$}} \arrow[rd, shorten <= -4pt, shorten >= -4pt] &  &  & \tikz{\node[draw, circle, inner sep=1pt]{$Q_{10}$}} \arrow[r] 
			\arrow[ld, shorten <= -4pt, shorten >= -4pt] & \tikz{\node[draw, circle, inner sep=1pt]{$Q_{11}$}} \arrow[rd, shorten <= -4pt, shorten >= -4pt]  &        &        & \tikz{\node[draw, circle, inner sep=1pt]{$Q_{12}$}} \arrow[ld, shorten <= -4pt, shorten >= -4pt] \\
			& \tikz{\node[draw, circle, inner sep=1pt]{$Q_{13}$}} \arrow[r] & \tikz{\node[draw, circle, inner sep=1pt]{$Q_{14}$}} & & & \tikz{\node[draw, circle, inner sep=1pt]{$Q_{15}$}} \arrow[r] & \tikz{\node[draw, circle, inner sep=1pt]{$Q_{16}$}} &
		\end{tikzcd}
	\end{adjustbox}
	\subfloat[\label{fig-quanarcht-a}]{\hspace{0.3\linewidth}}
	\subfloat[\label{fig-quanarcht-b}] {\hspace{0.5\linewidth}}
	\medskip
	\caption{\label{fig-quanarcht} The coupling graphs of (\textbf{a}) IBM QX20 Tokyo and (\textbf{b}) Rigetti 16Q-Aspen.}
\end{figure}
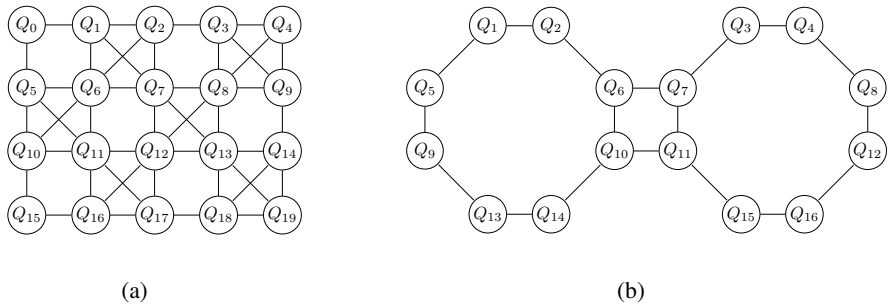

\begin{Definition}
	The coupling graph of a quantum computer $\mathbf{QM}$ is an undirected connected graph $(V,E_c)$, where $V$ is the set of qubits in $\mathbf{QM}$, and for any $q_i,q_j\in V$, $(q_i,q_j)\in E_c$ if and only if $q_i, q_j$ can be operated directly by a two-qubit gate in $\mathbf{QM}$.
\end{Definition}

Due to the nearest neighbor (NN) constraint, a quantum circuit need to be transformed into an equivalent and compliant one before executing in a quantum computer. This process is \textit{quantum circuit mapping} that contains two main steps: \textit{initial placement} and \textit{routing}.
The initial placement is to find a mapping from logical qubits in the quantum circuit to the physical qubits in the quantum computer in order to minimize non-NN interactions.
If two logical qubits that need to interact are not NN, we can move them by adding SWAP gates. Figure~\ref{Fig-qubitswap} shows a circuit using 2 SWAP gates to make $q_0$ and $q_3$ adjacent. The routing step is to find the best way to move logical qubits, i.e. reduces the number of SWAP gates as less as possible. 

	\begin{figure}[h]
	\centering
	\begin{adjustbox}{width=0.22\textwidth}
		\begin{quantikz}[row sep={0.5cm,between origins}]
			\lstick{$q_0$}  &  \qw 		&  \qw 		  &  \rstick{$q_0$} \qw \\		
			\lstick{$q_1$}  &  \qw		&  \swap{1}   &  \rstick{$q_3$} \qw \\
			\lstick{$q_2$}  &  \swap{1} &  \targX{}   &  \rstick{$q_1$} \qw \\
			\lstick{$q_3$}  &  \targX{} &  \qw  	  &  \rstick{$q_2$} \qw
		\end{quantikz}
	\end{adjustbox}
	\qquad
	\begin{adjustbox}{width=0.14\textwidth}
		\begin{tikzpicture}[baseline=1.5cm]
			\draw[step=1cm,gray,very thin] (0,0) grid (3,3);
			\node[fill, circle, inner sep=1.5pt] at (0,2) {};
			\node[fill, circle, inner sep=1.5pt] at (1,2) {};
			\node[] at (2,2) {\Large $\times$};
			\node[] at (2,1) {\Large $\times$};
			\node[] at (0.3,1.6) {\Large $q_0$};
			\node[] at (1.3,1.6) {\Large $q_1$};
			\node[] at (2.3,1.6) {\Large $q_2$};
			\node[] at (2.3,0.6) {\Large $q_3$};
		\end{tikzpicture}
	\end{adjustbox}
	$\Rightarrow$
	\begin{adjustbox}{width=0.14\textwidth}
		\begin{tikzpicture}[baseline=1.5cm]
			\draw[step=1cm,gray,very thin] (0,0) grid (3,3);
			\node[fill, circle, inner sep=1.5pt] at (0,2) {};
			\node[] at (1,2) {\Large $\times$};
			\node[] at (2,2) {\Large $\times$};
			\node[fill, circle, inner sep=1.5pt] at (2,1) {};
			\node[] at (0.3,1.6) {\Large $q_0$};
			\node[] at (1.3,1.6) {\Large $q_1$};
			\node[] at (2.3,1.6) {\Large $q_3$};
			\node[] at (2.3,0.6) {\Large $q_2$};
		\end{tikzpicture}
	\end{adjustbox}
	$\Rightarrow$
	\begin{adjustbox}{width=0.14\textwidth}
		\begin{tikzpicture}[baseline=1.5cm]
			\draw[step=1cm,gray,very thin] (0,0) grid (3,3);
			\node[fill, circle, inner sep=1.5pt] at (0,2) {};
			\node[fill, circle, inner sep=1.5pt] at (1,2) {};
			\node[fill, circle, inner sep=1.5pt] at (2,2) {};
			\node[fill, circle, inner sep=1.5pt] at (2,1) {};
			\node[] at (0.3,1.6) {\Large $q_0$};
			\node[] at (1.3,1.6) {\Large $q_3$};
			\node[] at (2.3,1.6) {\Large $q_1$};
			\node[] at (2.3,0.6) {\Large $q_2$};
		\end{tikzpicture}
	\end{adjustbox}
	\subfloat[\label{fig-qubitswap-a}]{\hspace{0.3\linewidth}}
	\subfloat[\label{fig-qubitswap-b}] {\hspace{0.5\linewidth}}
	\medskip
	\caption{\label{Fig-qubitswap} Move qubit states by SWAP gates. (\textbf{a}) A quantum circuit swapping the states of $q_1$, $q_2$ and $q_3$. (\textbf{b}) The movement of logical qubits in the quantum computer.}
\end{figure}
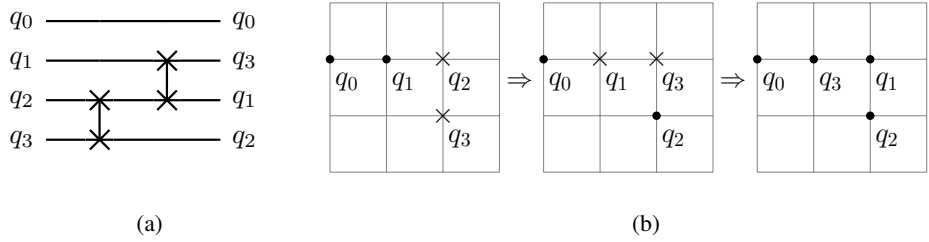

The quantum circuit mapping (QCM) problem is
\begin{description}
	\item[INPUT:] A quantum circuit $\Phi$, a coupling graph $G$, and a number $k$.
	\item[OUTPUT:] Yes, if $\Phi$ can be transformed to satisfy $G$'s NN constraint using at most $k$ SWAP gates; no, otherwise.
\end{description}

\begin{Proposition}\cite{Botea2018, Marcos2018}\label{prop-qctnpk0}
	The QCM problem is NP-complete.
\end{Proposition}
\begin{proof}
	To prove QCM is in NP, we can guess an initial mapping $\pi$ and $k$ swap operations and verify that whether $\Phi$ can be transformed to satisfy $G$'s NN constraint from $\pi$ by the $k$ swap operations. The verification can be done in polynomial time.
	The hardness of the problem is witnessed by the subgraph isomorphism problem that is NP-complete~\cite{garey1979}.
	That is, if $k=0$, then $\Phi$ satisfies $G$'s NN constraint iff the topology graph of $\Phi$ is a subgraph of $G$.
\end{proof}

\section{The complexity of QCM with parameters}\label{sec-complexity}
In this section, we shall study the complexity of QCM parameterized by the number of qubits, the depth of quantum circuits, the type of coupling graphs, and the number of swap operations. First an exact algorithm to compute the minimal number of swap operations needed is given for QCM. 

\subsection{An exact algorithm for QCM}\label{subsec-exactalgo}
Given a quantum circuit $\Phi=\langle Q, \Gamma \rangle$ and a coupling graph $G=(V,E_c)$, define
\[ 
g(\Phi,G)=
\begin{cases}
	k, & |Q| \leq |V|\\
	\infty, & |Q| > |V| 
\end{cases}
\]
where $k$ is the least number of SWAP gates needed to transform $\Phi$ to satisfy $G$'s NN  constraint. 
For simplicity, we always assume $|Q| \leq |V|$ in the remainder of this paper.
The \textit{distance} between two nodes $a,b$ in $G$ is the length of the shortest path from $a$ to $b$.
The \textit{diameter} of a graph $G$ is $\max\{d(a,b)\,|\,a,b\in V\}$, where $d(a,b)$ is the distance between $a$ and $b$. We use $dia(G)$ to denote the diameter of $G$. The following lemma gives an upper bound for $g(\Phi,G)$.

\begin{Lemma}\label{lem-max-swap}
	$g(\Phi,G)\leq t \cdotp (dia(G)-1)$, where $t$ is the number of two-qubit gates in $\Phi$.
\end{Lemma}
\begin{proof}
	Every two qubits mapped to $G$ can be made adjacent using at most $dia(G)-1$ swap operations.
	Suppose that there are $t$ two-qubit gates in $\Phi$, so $t\cdotp(dia(G)-1)$ SWAP gates are enough to transform $\Phi$ to satisfy $G$'s NN constraint.
\end{proof}

	In the following we give an upper bound for the number of subcircuits of $\Phi$.
	Let $\Phi$ be a quantum circuit with $n$ qubits and $l$ layers, we construct a black and white colored $n\times l$ grid $D$ ($n$ rows and $l$ columns) such that if qubit $q_i$ is operated by a gate in the $j$-th layer, then the node $(i,j)$ is black in $D$, and all nodes $(i,j')\,(j<j'\leq l)$ after it are also black in $D$. Then we compress the consecutive columns whose black nodes coincide in each row to one column to get a new grid $D'$. We call $D'$ the \textit{circuit type} of $\Phi$. 
	It is easily seen that there are at most $n$ columns in $D'$.
	In Figure~\ref{fig-cirtype} is an example that shows the procedure from a quantum circuit to its circuit type.

	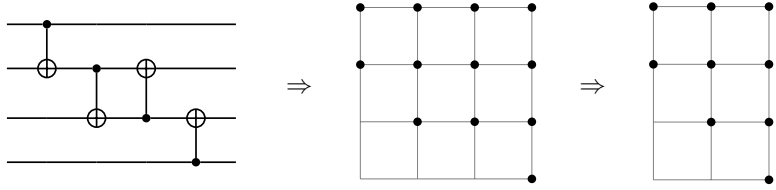
\begin{figure}[h]
		\centering
		\begin{adjustbox}{width=0.2\textwidth}
			\begin{quantikz}
				&   \ctrl{1} &  \qw  	& \qw	    & \qw  		&  \qw \\		
				& 	\targ{}  &  \ctrl{1}& \targ{}	& \qw  		&  \qw \\
				&   \qw		 &  \targ{} & \ctrl{-1} & \targ{}	&  \qw \\
				&   \qw	     &  \qw		& \qw  		& \ctrl{-1}   &  \qw
			\end{quantikz}
		\end{adjustbox}
		\hspace{0.4cm}
		$\Rightarrow$
		\hspace{0.4cm}
		\begin{adjustbox}{width=0.145\textwidth}
			\begin{tikzpicture}[baseline=1.5cm]
				\draw[step=1cm,gray,very thin] (0,0) grid (3,3);
				\node[fill, circle, inner sep=1.5pt] at (0,3) {};
				\node[fill, circle, inner sep=1.5pt] at (1,3) {};
				\node[fill, circle, inner sep=1.5pt] at (2,3) {};
				\node[fill, circle, inner sep=1.5pt] at (3,3) {};
				\node[fill, circle, inner sep=1.5pt] at (0,2) {};
				\node[fill, circle, inner sep=1.5pt] at (1,2) {};
				\node[fill, circle, inner sep=1.5pt] at (2,2) {};
				\node[fill, circle, inner sep=1.5pt] at (3,2) {};
				\node[fill, circle, inner sep=1.5pt] at (1,1) {};
				\node[fill, circle, inner sep=1.5pt] at (2,1) {};
				\node[fill, circle, inner sep=1.5pt] at (3,0) {};
				\node[fill, circle, inner sep=1.5pt] at (3,1) {};
			\end{tikzpicture}
		\end{adjustbox}	
		\hspace{0.4cm}
		$\Rightarrow$
		\hspace{0.4cm}
		\begin{adjustbox}{width=0.1\textwidth}
			\begin{tikzpicture}[baseline=1.5cm]
				\draw[step=1cm,gray,very thin] (0,0) grid (2,3);
				\node[fill, circle, inner sep=1.5pt] at (0,3) {};
				\node[fill, circle, inner sep=1.5pt] at (1,3) {};
				\node[fill, circle, inner sep=1.5pt] at (2,3) {};
				\node[fill, circle, inner sep=1.5pt] at (0,2) {};
				\node[fill, circle, inner sep=1.5pt] at (1,2) {};
				\node[fill, circle, inner sep=1.5pt] at (2,2) {};
				\node[fill, circle, inner sep=1.5pt] at (1,1) {};
				\node[fill, circle, inner sep=1.5pt] at (2,1) {};
				\node[fill, circle, inner sep=1.5pt] at (2,0) {};
			\end{tikzpicture}	
		\end{adjustbox}
		\caption{\label{fig-cirtype} From a quantum circuit to its circuit type.}
	\end{figure}

	Let $n\geq 1$ and $T_n$ be a set of colored $n\times m$  grids ($1\leq m \leq n$) such that a grid is in $T_n$ iff it satisfies:
	\begin{itemize}
	\item At least one node is black in the first column.
	\item If the node $(i,j)$ is black, then the nodes $(i,j')\,(j< j' \leq m)$ are also black.
	\item For every $1< j \leq m$, there is $1\leq i \leq n$ such that $(i,j)$ is black, $(i,j-1)$ is white.
	\end{itemize}
	The grid in $T_n$ is monotone in the sense that the black nodes in every column grow from left to right. As an example, Figure~\ref{Fig-grid} shows four grids in $T_4$. 

	\begin{figure}[h]
	\centering
	\begin{adjustbox}{width=0.55\textwidth}
	\begin{tikzpicture}
	\draw[step=1cm,gray,very thin] (0,0) grid (3,3);
	\node[fill, circle, inner sep=1.5pt] at (0,0) {};
	\node[fill, circle, inner sep=1.5pt] at (1,0) {};
	\node[fill, circle, inner sep=1.5pt] at (2,0) {};
	\node[fill, circle, inner sep=1.5pt] at (3,0) {};
	\node[fill, circle, inner sep=1.5pt] at (1,1) {};
	\node[fill, circle, inner sep=1.5pt] at (2,1) {};
	\node[fill, circle, inner sep=1.5pt] at (3,1) {};
	\node[fill, circle, inner sep=1.5pt] at (2,2) {};
	\node[fill, circle, inner sep=1.5pt] at (3,2) {};
	\node[fill, circle, inner sep=1.5pt] at (3,3) {};
	\end{tikzpicture}
	\hspace{2cm}
	\begin{tikzpicture}
	\draw[step=1cm,gray,very thin] (0,0) grid (2,3);
	\node[fill, circle, inner sep=1.5pt] at (0,1) {};
	\node[fill, circle, inner sep=1.5pt] at (1,2) {};
	\node[fill, circle, inner sep=1.5pt] at (2,0) {};
	\node[fill, circle, inner sep=1.5pt] at (1,1) {};
	\node[fill, circle, inner sep=1.5pt] at (2,1) {};
	\node[fill, circle, inner sep=1.5pt] at (2,2) {};
	\node[fill, circle, inner sep=1.5pt] at (2,3) {};
	\end{tikzpicture}
	\hspace{2cm}
	\begin{tikzpicture}
	\draw[step=1cm,gray,very thin] (0,0) grid (1,3);
	\node[fill, circle, inner sep=1.5pt] at (0,1) {};
	\node[fill, circle, inner sep=1.5pt] at (0,2) {};
	\node[fill, circle, inner sep=1.5pt] at (1,2) {};
	\node[fill, circle, inner sep=1.5pt] at (1,3) {};
	\node[fill, circle, inner sep=1.5pt] at (1,1) {};
	\end{tikzpicture}
	\hspace{2cm}
	\begin{tikzpicture}
	\draw[step=1cm,gray,very thin] (0,0) grid (0,3);
	\node[fill, circle, inner sep=1.5pt] at (0,1) {};
	\node[fill, circle, inner sep=1.5pt] at (0,2) {};
	\node[fill, circle, inner sep=1.5pt] at (0,0) {};
	\node[fill, circle, inner sep=1.5pt] at (0,3) {};
	\end{tikzpicture}
	\end{adjustbox}
	\caption{\label{Fig-grid} Four grids in $T_4$.}
	\end{figure}
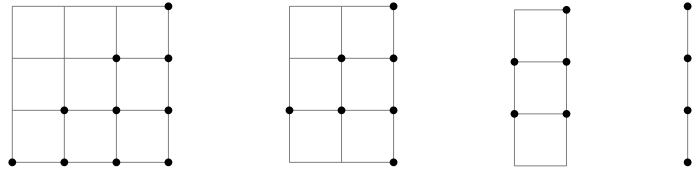

\begin{Lemma}
	 If $\Phi$ is a quantum circuit with $n$ qubits and $l$ layers, then there are at most $n!\cdot (2l)^n$ subcircuits of $\Phi$.
\end{Lemma}
\begin{proof}
	Every subcircuit $\Phi'$ of $\Phi$ corresponds to an $n\times l$ grid and its circuit type is an $n \times m$ grid ($m\leq n$). 
	Let $m\leq \min\{n,l\}$, there are at most $l \choose m $ possibilities to expand an $n\times m$ grid to an $n\times l$ grid by duplicating columns. 
	Since ${l \choose m} \leq l^m \leq l^n$, there are at most $l^n$ subcircuits of $\Phi$ that have the same circuit type.
	
	Note that $T_n$ contains the circuit types of all quantum circuits of $n$ qubits.
	The number of $n\times n$ grids in $T_n$ is at most $n!$. An $n\times m\,(1\leq m < n)$ grid in $T_n$ can be obtained by removing some columns from an $n\times n$ grid, so $|T_n| \leq n!\cdot 2^n$. 
	Therefore, the total number of subcircuits of $\Phi$ is at most $n!\cdot 2^n \cdot l^n = n!\cdot (2l)^n$. 
\end{proof}

\begin{Corollary}\label{coro-cardinal}
	For any minimized subcircuit set $S$ of $\Phi$, $|S|\leq n!\cdot (2l)^n$.
\end{Corollary}

\begin{Proposition}\label{prop-compute-f}
	Given a quantum circuit $\Phi$ with $n$ qubits and $t$ gates, and a coupling graph $G=(V,E_c)$ with $b$ nodes, $g(\Phi,G)$ is computable in $O(b^{3n+3} t n + (n!)^2 4^n t^{2n+2} b^{2n+3})$ time.
\end{Proposition}

\begin{proof}
	There are at most $\frac{b!}{(b-n)!}$ mappings from the qubits in $\Phi$ to the nodes in $G$. We compute the minimal number of SWAP gates needed for each initial mapping, then $g(\Phi,G)$ is the smallest one among them.
	We do a breadth-first search in Algorithm~\ref{algo-compute-fun}. In each iteration, we consider all possible swap operations, and transform the quantum circuit according to them. 
	We reduce $\Phi$ after each swap operation by removing all gates that are already nearest neighbored and do not depend on any other gates to get a subcircuit.
	To prohibit search space explosion, only the minimal subcircuits are saved. The procedure is repeated until an empty circuit occurs. The outline of Algorithm~\ref{algo-compute-fun} is in the following.
	
	\begin{description}[leftmargin=10pt]
		\item[Step 1] Given an initial mapping $\pi_0$, reduce $\Phi$ by removing NN-gates under $\pi_0$ to get a subcircuit $\Phi'$. Set $\Pi=\{(\pi_0,\Phi')\}$ and $\Pi'=\emptyset$. 
		\item[Step 2] For each $(\pi,\Phi_1,\dots,\Phi_d)\in \Pi$ do \textbf{Step 3} until an empty circuit occurs.
		\item[Step 3] For each edge in $E_c$, make a swap operation on the qubits attached to the edge to get a new mapping $\pi'$.
		\begin{description}[leftmargin=10pt]
			\item[Step 3.1] Reduce each $\Phi_i\,(1\leq i \leq d)$ under $\pi'$ to get a set $\{\Phi_1',\dots,\Phi_{d}'\}$, and minimize it to get a new set $\{\Phi_1'',\dots,\Phi_{d'}''\}$ that contains only the minimal subcircuits.
			\item[Step 3.2]  If there is a tuple $\mathbf{Tu} = (\pi',\Psi_1,\dots,\Psi_{e}) \in \Pi'$ with the same mapping $\pi'$, then minimize the set $\{\Psi_1,\dots,\Psi_{e},\Phi_1'',\dots,\Phi_{d'}''\}$ to get a new set $\{\Psi_1',\dots,\Psi_{e'}'\}$, and replace $\mathbf{Tu}$ with $(\pi',\Psi_1',\dots,\Psi_{e'}')$ in $\Pi'$.
			Otherwise, put $(\pi',\Phi_1'',\dots,\Phi_{d'}'')$ in $\Pi'$.
		\end{description}
		\item[Step 4] If \textbf{Step 2} is finished and no empty circuit occurs, then replace $\Pi$ by $ \Pi'$ and set $\Pi' =\emptyset$, and repeat \textbf{Step 2} again. Otherwise, return the number of repeated times.
	\end{description}	
	
	\LinesNumbered
	\begin{algorithm}[h]
		\KwIn{A quantum circuit $\Phi$ and a coupling graph $G=(V,E_c)$.}
		\KwOut{The value of $g(\Phi,G)$.}
		\ForEach{initial mapping $\pi_0$} {
			$\Phi'\leftarrow$ reduce $\Phi$ under $\pi_0$, $\Pi\leftarrow\{(\pi_0,\Phi')\},\Pi'\leftarrow\emptyset$, $N_{\pi_0} \leftarrow 0$\;
			\leIf{$\Phi'=\emptyset$}{$Find \leftarrow Yes $}{$Find \leftarrow No$}
			\While{$Find= No$}{
				\ForEach{$(\pi,\Phi_1,\dots,\Phi_d)\in \Pi$}{
					\ForEach{edge $y \in E_c$}{
						Swap the logical qubits attached to $y$ to get a new mapping $\pi'$\;
						$(\Phi_1',\dots,\Phi_{d}')\leftarrow$ reduce $(\Phi_1,\dots,\Phi_d)$ under $\pi'$\;
						$(\Phi_1'',\dots,\Phi_{d'}'')\leftarrow$ minimize $(\Phi_1',\dots,\Phi_{d}')$\;
						\lIf{$\exists \Phi_i''=\emptyset$}{$Find\leftarrow Yes$}
						\eIf{$\exists (\pi',\Psi_1,\dots,\Psi_{e})\in \Pi'$}
						{$\{\Psi_1',\dots,\Psi_{e'}'\}\leftarrow$ minimize $\{\Psi_1,\dots,\Psi_{e},\Phi_1'',\dots,\Phi_{d'}''\}$\;
							Replace $(\pi',\Psi_1,\dots,\Psi_{e})$ with $(\pi',\Psi_1',\dots,\Psi_{e'}')$ in $\Pi'$\;}{
							$\Pi'\leftarrow \Pi'\cup (\pi',\Phi_1'',\dots,\Phi_{d'}'')$\;
						}
					}
				}
				$\Pi \leftarrow \Pi', \Pi'\leftarrow\emptyset, N_{\pi_0} \leftarrow N_{\pi_0} + 1$\;
			}
		}
		$N\leftarrow \min\{N_{\pi_0} |\pi_0 \text{ is an initial mapping}\}$\;
		\Return N\;
		
		\caption{An exact algorithm for computing $g(\Phi,G)$.}
		\label{algo-compute-fun}
	\end{algorithm}

	In \textbf{Step 1}, it takes $O(t)$ time to reduce $\Phi$. Each mapping occurs at most once in $\Pi$, so $|\Pi|\leq \frac{b!}{(b-n)!}$, and \textbf{Step 2} repeats at most $\frac{b!}{(b-n)!}$ times.
	\textbf{Step 3} repeats at most $|E_c|$ times.
	By Lemma~\ref{lem-max-swap}, \textbf{Step 4} repeats at most $t \cdotp (dia(G)-1)$ times.
	Note that the cardinality of $\{\Phi_1,\dots,\Phi_d\}$ is at most $n!\cdot (2l)^n$ for any tuple $(\pi,\Phi_1,\dots,\Phi_d)\in \Pi$ by Corollary~\ref{coro-cardinal}, where $l$ is the number of layers of $\Phi$.
	For the time complexity of \textbf{Step 3}, in \textbf{Step 3.1}, it takes $O(n!\cdot (2l)^n\cdotp t)$ time to reduce $\Phi_1,\dots,\Phi_d$, and $O((n!)^2\cdot (2l)^{2n}\cdotp t)$ time to minimize the new subcircuit set. Thus, \textbf{Step 3.1} takes $O((n!)^2\cdot (2l)^{2n}\cdotp t)$ time. 
	In \textbf{Step 3.2}, it takes $O(\frac{n\cdot b!}{(b-n)!})$ time to check whether there is a tuple in $\Pi'$ with the same mapping, and $O((n!)^2\cdot (2l)^{2n}\cdotp t)$ time to minimize them if there is one. Thus, \textbf{Step 3.2} takes $O(\frac{n\cdot b!}{(b-n)!} + (n!)^2\cdot (2l)^{2n}\cdotp t)$ time.
	The total running time of Algorithm~\ref{algo-compute-fun} is
	\[
	\frac{b!}{(b-n)!} \cdotp \left( t \cdotp (dia(G)-1) \cdotp \frac{b!}{(b-n)!} \cdotp |E_c| \cdotp O(\frac{n\cdot b!}{(b-n)!} + (n!)^2\cdot (2l)^{2n}\cdotp t)  + O(t)\right).
	\]
	Since $l\leq t$, $\frac{b!}{(b-n)!}\leq b^n$, $dia(G)-1 \leq b$ and $|E_c|\leq b^2$, we see that $g(\Phi,G)$ is computable in $O(b^{3n+3} t n + (n!)^2 4^n t^{2n+2} b^{2n+3})$ time.
\end{proof}

\begin{Proposition}
	If the number of qubits in $\Phi$ is fixed to a constant, then $g(\Phi,G)$ is computable in polynomial time.
\end{Proposition}
\begin{proof}
	By Proposition~\ref{prop-compute-f},  if the number of qubits in $\Phi$ is fixed to a constant $c$, then $g(\Phi,G)$ is computable in $O(t^{2c+2}b^{3c+3})$ time. This completes the proof.
\end{proof}

\begin{Corollary}
	If the number of nodes in the coupling graph $G$ is bounded by a constant, or $G$ is a fixed graph, then $g(\Phi,G)$ is computable in polynomial time.
\end{Corollary}

\subsection{The complexity of QCM with fixed number of qubits}\label{subsec-qubitnum}

The QCM problem is computable in polynomial time if the number of qubits is bounded. We further show that it is NL-complete. Before stating the result to be proved, we give two lemmas in the following. The undirected shortest path (USP) problem is 
\begin{description}
	\item[INPUT:] An undirected graph $G$, two nodes $\mathbf{s},\mathbf{t}$ in $G$, and a number $k$.
	\item[OUTPUT:] Yes, if the distance between $\mathbf{s}$ and $\mathbf{t}$ is at most $k$; no, otherwise.
\end{description}
The USP problem is NL-complete~\cite{tantau2005}. We shall show that it is still NL-complete for bounded degree graphs. 
For undirected graphs, the \textit{degree} of a node is the number of edges that are incident to it. The degree of an undirected graph is the maximum of its nodes' degrees.
For directed graphs, the \textit{outdegree} (\textit{indegree}) of a node is the number of outcoming (incoming) edges that are incident to it. The outdegree (indegree) of a directed graph is the maximum of its nodes' outdegrees (indegrees).

\begin{Lemma}\label{lem-shortestpath}
	The USP problem for graphs with maximum degree 3 is NL-complete. 
\end{Lemma}
\begin{proof}
	Obviously, the problem is in NL.
	To prove the hardness, we make a reduction from the ST-connectivity problem that is NL-complete~\cite{sipser1996}.  Given an instance $(H,\mathbf{s},\mathbf{t})$ of the ST-connectivity problem, where
	$H$ is a directed graph and $\mathbf{s},\mathbf{t}$ are two nodes in $H$, it asks whether there is a directed path from $\mathbf{s}$ to $\mathbf{t}$. 
	
	Let $a$ be a node in $H$. Denote by $outdegree(a)$ ($indegree(a)$) the outdegree (indegree) of $a$, $deg(H)$ the maximum of the outdegree and indegree of $H$. Define $d = \lceil \log(deg(H)+1) \rceil$, and $T^d_{out}(a)$ ($T^d_{in}(a)$) to be a binary tree of height $d$ such that the number of leaf nodes is $outdegree(a)+1$ ($indegree(a)+1$), and all leaf nodes have depth $d$.
	Suppose that $H$ contains $m$ nodes $a_1(=\mathbf{s}),\dots,a_m(=\mathbf{t})$. Define an undirected graph $G$ to contain $m$ rows and each row to contain $m$ nodes. 
	Let $a_{(i,j)}$ be the $j$-th node in the $i$-th row ($1\leq i \leq m, 1\leq j \leq m$).	
	For every node $a_{(1,j)}$ ($a_{(m,j)}$) where $1\leq j \leq m$, we attach a tree $T^d_{out}(a_{(1,j)})$ ($T^d_{in}(a_{(m,j)})$) that is an isomorphic copy of $T^d_{out}(a_j)$ ($T^d_{in}(a_j)$) to it, i.e, connect it to the root of the tree. 
	For every other node $a_{(i,j)}$ ($1< i < m, 1\leq j \leq m$), we attach two trees $T^d_{out}(a_{(i,j)})$ and $T^d_{in}(a_{(i,j)})$ that are isomorphic copies of $T^d_{out}(a_j)$ and $T^d_{in}(a_j)$ respectively to it.
	For $1\leq i <m$, we merge a leaf node of $T^d_{out}(a_{(i,j)})$ and a leaf node of $T^d_{in}(a_{(i+1,j)})$, and if there is an edge form $a_j$ to $a_{j'}$ in $H$, then we merge a leaf node of $T^d_{out}(a_{(i,j)})$ and a leaf node of $T^d_{in}(a_{(i+1,j')})$.
	Using the trees $T^d_{out}$ and $T^d_{in}$ as the connecting gadget guarantees that the degree of every node in $G$ is at most 3, and the distance between two nodes $a_{(i,j)}$, $a_{(i+1,j')}$ in two consecutive rows is $2d+2$ iff $j=j'$ or there is an edge from $a_j$ to $a_{j'}$ in $H$.
	
	If there is a directed path $\mathbf{s}\rightarrow a_{j}\rightarrow a_{j'}\rightarrow\cdots \rightarrow a_{j''}\rightarrow\mathbf{t}$ in $H$, then there is a path $a_{(1,1)}-\cdots -a_{(2,j)}-\cdots - a_{(3,j')}-\cdots -a_{(i,j'')}-\cdots -a_{(i+1,m)}-\cdots -a_{(m,m)}$ of length $(m-1)(2d+2)$ in $G$. Conversely, if there is a path form $a_{(1,1)}$ to $a_{(m,m)}$ with length at most $(m-1)(2d+2)$, we can deduce that the length is exactly $(m-1)(2d+2)$, since the distance between a node in the first row and a node in the last row is at least $(m-1)(2d+2)$ by the construction of $G$. A path from $\mathbf{s}$ to $\mathbf{t}$ in $H$ can be found by the path form $a_{(1,1)}$ to $a_{(m,m)}$ with length $(m-1)(2d+2)$ in $G$.

	So we can reduce $(H,\mathbf{s},\mathbf{t})$ to the instance $(G,a_{(1,1)},a_{(m,m)},(m-1)(2d+2))$ such that there is a directed path from $\mathbf{s}$ to $\mathbf{t}$ iff the distance between $a_{(1,1)}$ and $a_{(m,m)}$ is at most $(m-1)(2d+2)$ in $G$, where the degree of $G$ is at most 3.
\end{proof}

When transforming $\Phi$ over a coupling graph $G$ with the initial mapping $\pi$, we define $lw(\Phi,G,\pi)$ to be the least number $h$ of swap operations $\alpha_1,\dots,\alpha_h$ such that
\begin{itemize}
	\item if $\Psi$ is a subcircuit obtained by reducing $\Phi$ after these $h$ swap operations, then it has the same number of qubits as that of $\Phi$, and
	\item there is a swap operation $\alpha$ such that if $\Psi'$ is obtained by reducing $\Phi$ after the swap operations $\alpha_1,\dots,\alpha_h,\alpha$, then the number of qubits in $\Psi'$ is strictly less than that of $\Phi$.
\end{itemize}
Set $lw(\Phi,G)=\min\{lw(\Phi,G,\pi)\mid \pi \text{ is an initial mapping}\}$.

\begin{Lemma}\label{lem-least-swap}
	Suppose that $\Phi_1$ and $\Phi_2$ have the same number of qubits, and $\Phi$ is obtained by concatenating $\Phi_1$ and $\Phi_2$, then $lw(\Phi,G)\geq lw(\Phi_1,G)+lw(\Phi_2,G)$.
\end{Lemma}
\begin{proof}
	Given an arbitrary initial mapping $\pi$, $lw(\Phi_1,G)$ is the least number of swap operations
	from which a subcircuit of $\Phi$ containing some gates in $\Phi_1$ and all gates in $\Phi_2$ can be obtained. Obviously, any $lw(\Phi_2,G)$ swap operations in the following will produce a subcircuit that has the same number of qubits as that of $\Phi$. Therefore, $lw(\Phi,G)\geq lw(\Phi_1,G)+lw(\Phi_2,G)$.
\end{proof}

We will denote by $\Phi^{*h}$ the quantum circuit constituted of $h$ consecutive copies of $\Phi$. It is easily seen that $g(\Phi^{*h},G)\geq lw(\Phi^{*h},G) \geq h\cdot lw(\Phi,G)$ by the above lemma.

\begin{Proposition}\label{prop-nl-comple}
	If the number of qubits in $\Phi$ is fixed, then the QCM problem is NL-complete.
\end{Proposition}
\begin{proof}
	We first give a nondeterministic algorithm for the QCM problem in Algorithm~\ref{algo-compute-logspace}. It works as follows.
	\begin{description}[leftmargin=10pt]
		\item[Step 1] For each initial mapping $\pi_0$ do the following.
		\item[Step 2] Reduce $\Phi$ under $\pi_0$ to get a subcircuit of $\Phi$.
		\item[Step 3] Nondeterministically guess a swap operation, change to a new mapping $\pi$ and reduce the previous subcircuit to a new subcircuit under $\pi$.
		\item[Step 4] Repeat \textbf{Step 3} at most $k$ times. Accept if an empty circuit occurs, reject otherwise.
	\end{description}

	\LinesNumbered
	\begin{algorithm}[h]
		\KwIn{A quantum circuit $\Phi$, a coupling graph $G$, and a number $k$.}
		\KwOut{Yes, if $k$ SWAP gates are enough to tranform $\Phi$; no otherwise.}
		\ForEach{initial mapping $\pi_0$} {
			Set $N = k$ and $\pi = \pi_0$\;
			$\Psi \leftarrow$ reduce $\Phi$ under $\pi$\;
			\While{$N \geq 1$}{
				\lIf{$\Psi$ is empty}{break and \Return yes}
				Nondeterministically guess a swap operation and update $\pi$\;
				$\Psi \leftarrow$ reduce $\Psi$ under $\pi$\;
				$N = N-1$\;	
			}
		}
		\caption{A nondeterministic algorithm for QCM.}
		\label{algo-compute-logspace}
	\end{algorithm} 
	
	Assume that $\Phi$ has $n$ qubits and $l$ layers. Let $\Psi$ be a subcircuit of $\Phi$, by the definition of circuit type, $\Psi$ corresponds to an $n\times l$ grid $D$ and can be compressed to an $n\times m$ ($m\leq n$) grid $D'$ that is the circuit type of $\Psi$. Define $T_{\Psi}=[(l_1,Q_1),\dots,(l_m,Q_m)]$, where $1\leq l_i\leq l$ and the $D$'s $l_i$-th column is the first one that is the same as the $D'$'s $i$-th column, and $Q_i$ is the set of nodes that are black in the $i$-th column of $D'$. 
	$\Psi$ can be easily recovered from $\Phi$ and $T_{\Psi}$, we shall use $T_{\Psi}$ instead of $\Psi$ in Algorithm~\ref{algo-compute-logspace} to save space.

	We now consider the space complexity of Algorithm~\ref{algo-compute-logspace}. 
	Suppose that there are $b$ nodes in $G$.
	The number $N$ and the mapping $\pi$ use $O(\log k)$ and $O(n\log b)$ space to store, respectively. Since $\Psi$ is replaced by $T_{\Psi}$, the space used is $O(n(\log l + n\log n))$. And it takes $O(\log b)$ space to store a swap operation. So the total space used by the algorithm is $O(\log k + n(\log b +\log l + n\log n))$. By Lemma~\ref{lem-max-swap}, we can assume that $k \leq n\cdot l\cdot (dia(G)-1)$. Since $n$ is a constant for all inputs, Algorithm~\ref{algo-compute-logspace} runs in $O(\log (b+l))$ space.
	Hence, the QCM problem is in NL when the number of qubits in $\Phi$ is bounded by a constant.
	
	To prove the hardness, we reduce the USP problem for graphs with maximum degree 3, which is NL-complete by Lemma~\ref{lem-shortestpath}, to the QCM problem where the quantum circuit $\Phi$ has $10$ qubits. 
	Let $\Phi^5_{clique}$ be a quantum circuit with 5 qubits and every pair of qubits are operated by a CNOT gate (see Figure~\ref{fig-cliquecircuit-a}). The topology graph of $\Phi^5_{clique}$ is a clique of size 5. 
	
	\begin{figure}[h]
		\centering
		\begin{adjustbox}{width=0.3\textwidth}
			\begin{quantikz}
				& \ctrl{1} & \ctrl{2} & \ctrl{3} & \ctrl{4} & \qw  	   & \qw 		& \qw 			& \qw	    &  \qw \\		
				& \targ{}  & \qw 	  & \qw 	 & \qw 		& \ctrl{1} & \ctrl{2}   & \ctrl{3} 		& \qw	    &  \qw \\
				& \ctrl{1} & \targ{}  & \qw 	 & \qw 		& \targ{}  & \qw 		& \qw 			& \ctrl{2} 	&  \qw \\
				& \targ{}  & \ctrl{1} & \targ{}  & \qw 		& \qw	   & \targ{}	& \qw 			& \qw  		&  \qw \\
				& \qw	   & \targ{}  & \qw 	 & \targ{}  & \qw	   & \qw		& \targ{} 		& \targ{}  	&  \qw
			\end{quantikz}
		\end{adjustbox}
		\hspace{1.5cm}
		\begin{adjustbox}{width=0.16\textwidth}
			\begin{tikzpicture}[baseline=0cm]
				\graph {subgraph K_n [n=5,clockwise,radius=2cm]};
			\end{tikzpicture}
		\end{adjustbox}
		\subfloat[\label{fig-cliquecircuit-a}]{\hspace{0.4\linewidth}}
		\subfloat[\label{fig-cliquecircuit-b}]{\hspace{0.3\linewidth}}
		\caption{\label{fig-cliquecircuit} The illustration of (\textbf{a}) the quantum circuit $\Phi^5$, and (\textbf{b}) its topology graph.}
	\end{figure}
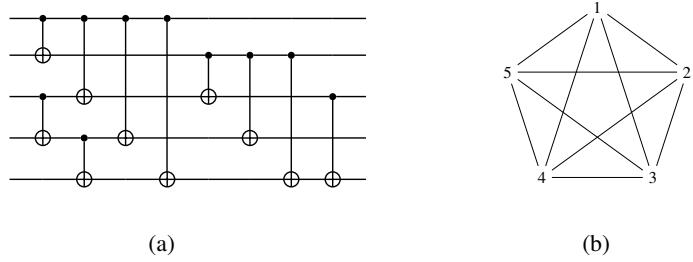
	
	Given an instance $(G',\mathbf{s},\mathbf{t},k)$ of the USP problem where $G'$ is of maximum degree 3, we define an instance $(\Phi,G,k-1)$ of the QCM problem as follows.
	\begin{itemize}
		\item $G$ is obtained by attaching two cliques of size 4 to $\mathbf{s}$ and $\mathbf{t}$ in $G'$, respectively, i.e., there are edges connecting $\mathbf{s}$ (and $\mathbf{t}$) to all nodes of the clique.
		\item $\Phi$ is constituted of two parallel copies of $(\Phi^5_{clique})^{*h}$ followed by a CNOT gate that operates on the first qubits of the two copies, where $h$ is the number of nodes in $G'$.
	\end{itemize}
If the distance between $\mathbf{s}$ and $\mathbf{t}$ is at most $k$, we can map the qubits of the two copies $(\Phi^5_{clique})^{*h}$ to the two cliques containing $\mathbf{s}$ and $\mathbf{t}$ in $G$, respectively. Hence, only the last CNOT gate need to be made nearest neighbor when executing it. This can be realized by moving $\mathbf{s}$ or $\mathbf{t}$ along the shortest path between them, where at most $k-1$ SWAP gates are enough.
On the other hand, suppose that $\Phi$ can be transformed to satisfy $G$'s NN constraint using at most $k-1$ SWAP gates. 
If the qubits in $(\Phi^5_{clique})^{*h}$ are mapped to the clique containing $\mathbf{s}$ or $\mathbf{t}$, it will use $k'-1$ SWAP gates to  make the last CNOT gate nearest neighbor, where $k'$ is the distance between $\mathbf{s}$ and $\mathbf{t}$. Otherwise, if $(\Phi^5_{clique})^{*h}$ is not implemented on the clique of size 5, then it needs at least $h$ SWAP gates by the construction of $G$ and Lemma~\ref{lem-least-swap}.
Since $k'\leq h$ and at most $k-1$ SWAP gates are enough to transform $\Phi$ over $G$, we have $k'-1\leq k-1$, which implies the distance between $\mathbf{s}$ and $\mathbf{t}$ in $G'$ is at most $k$.
\end{proof}

\subsection{Fixed-parameter complexity of QCM}\label{subsec-paracomplxy}
Though QCM is computable in polynomial time when the number of qubits is bounded, the algorithm is impractical due to the large hidden constants. 
We make an experiment of algorithm~\ref{algo-compute-fun} on a linear coupling graph and random quantum circuits with 3 and 4 qubits. Form Table~\ref{tab-algo1} and Figure~\ref{fig-experiment} we see that the running time of the quantum circuits of 4 qubits grows much faster than that of the quantum circuits of 3 qubits. This implies that the QCM problem is hard to solve if the number of qubits is fixed to a large number.
Actually, we show that QCM is fixed-parameter intractable parameterized by the number of qubits in the quantum circuit.
\begin{figure}[h!]
	\centering
	\begin{adjustbox}{width=0.8\textwidth}
		\begin{tikzpicture}
		\begin{axis}[
		title= {The value of $g(\Phi, G)$},
		xlabel={Number of gates in the circuit},
		ylabel={Number of SWAP gates},
		legend style={at={(0.25,0.85)}, anchor=center},
		]
		\addplot table {data11.dat};
		\addplot table {data2.dat};
		\legend{3 qubits, 4 qubits}
		\end{axis}
		\end{tikzpicture}
		\hspace{1cm}
		\begin{tikzpicture}
		\begin{axis}[
		title={The running time of Algorithm~\ref{algo-compute-fun}},
		xlabel={Number of gates in the circuit},
		ylabel={Time in seconds},
		legend style={at={(0.25,0.85)}, anchor=center},
		]
		\addplot table {data31.dat};
		\addplot table {data4.dat};
		\legend{3 qubits, 4 qubits}
		\end{axis}
		\end{tikzpicture}
	\end{adjustbox}
	\caption{An illustration of the number of SWAP gates and running time in Table~\ref{tab-algo1}.
	\label{fig-experiment}}
\end{figure}
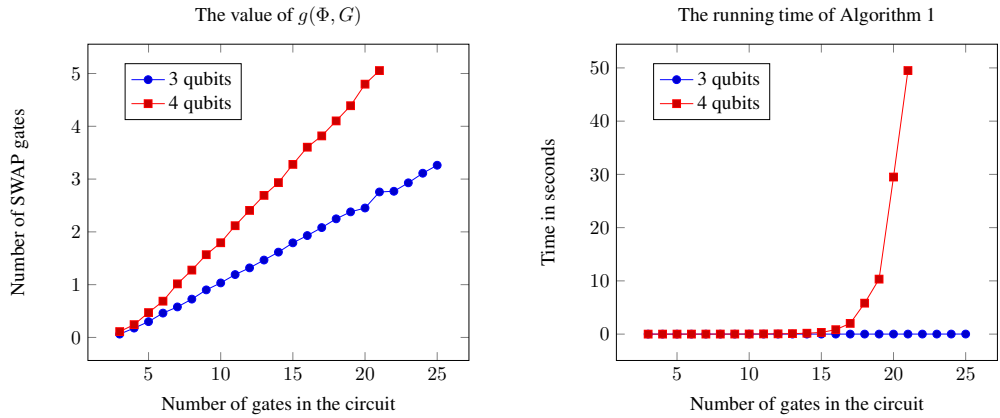

\begin{table}[h]
	\centering
	\caption{The minimal number of SWAP gates and running time of Algorithm~\ref{algo-compute-fun} on a linear coupling graph and random quantum circuits with 3 and 4 qubits, respectively. \label{tab-algo1}}
	\footnotesize
	\begin{tabular}{|c|c|c|c|c|}\hline
		\multirow{2}{*}{\diagbox{$gt$}{$qt$}} & \multicolumn{2}{c|}{3 qubits} & \multicolumn{2}{c|}{4 qubits} \\ \cline{2-5}
		&  $gt_{add}$  &  $time$  &  $gt_{add}$  &  $time$  \\ \hline
		3  &  0.063 & $1.95\times 10^{-4}$  & 0.111 &   $5.59 \times 10^{-4}$  \\ \hline
		4  &  0.177 & $2.91\times 10^{-4}$  & 0.243 &   $1.07\times 10^{-3}$ \\ \hline
		5  &  0.299 & $4.20\times 10^{-4}$  & 0.471 &   $2.03\times 10^{-3}$ \\ \hline
		6  &  0.461 & $5.76\times 10^{-4}$  & 0.687 &   $3.27\times 10^{-3}$ \\ \hline
		7  &  0.578 & $7.06\times 10^{-4}$  & 1.015 &   $6.62\times 10^{-3}$ \\ \hline
		8  &  0.726 & $8.96\times 10^{-4}$  & 1.276 &   $1.10\times 10^{-2}$ \\ \hline
		9  &  0.901 & $1.12\times 10^{-3}$  & 1.568 &   $1.46\times 10^{-2}$ \\ \hline
		10  & 1.033 & $1.31\times 10^{-3}$  & 1.794 &   $2.05\times 10^{-2}$ \\ \hline
		11  & 1.191 & $1.58\times 10^{-3}$  & 2.117 &   $4.18\times 10^{-2}$ \\ \hline
		12  & 1.318 & $1.86\times 10^{-3}$  & 2.406 &   $5.29\times 10^{-2}$ \\ \hline
		13  & 1.466 & $2.22\times 10^{-3}$  & 2.691 &   $9.34\times 10^{-2}$ \\ \hline
		14  & 1.616 & $2.63\times 10^{-3}$  & 2.933 &   $1.67\times 10^{-1}$ \\ \hline		  
		15  & 1.792 & $3.08\times 10^{-3}$  & 3.278 &   $3.15\times 10^{-1}$ \\ \hline
		16  & 1.930 & $3.62\times 10^{-3}$  & 3.604 &   $8.37\times 10^{-1}$ \\ \hline
		17  & 2.081 & $4.14\times 10^{-3}$  & 3.818 &   $2.013$ \\ \hline
		18  & 2.247 & $4.92\times 10^{-3}$  & 4.102 &   $5.814$ \\ \hline
		19  & 2.379 & $5.69\times 10^{-3}$  & 4.390 &   $10.33$ \\ \hline
		20  & 2.452 & $6.23\times 10^{-3}$  & 4.798 &   $29.50$ \\ \hline
		21  & 2.756 & $8.30\times 10^{-3}$  & 5.056 &   $49.50$ \\ \hline
		22  & 2.769 & $8.35\times 10^{-3}$  & - &  -  \\ \hline
		23  & 2.929 & $9.58\times 10^{-3}$  & - &  - \\ \hline
		24  & 3.111 & $1.14\times 10^{-2}$  & - &  - \\ \hline
		25  & 3.262 & $1.36\times 10^{-2}$  & - &  - \\ \hline
	\end{tabular}
	\medskip
	
	$gt$: the number of gates in the circuit; $qt$: the number of qubits in the circuit; $gt_{add}$: the average number of SWAP gates added; $time$: the average running time in seconds.
\end{table}

A problem is \textit{fixed-parameter tractable} if it is computable in time $f(p)\cdotp |x|^{O(1)}$, where $|x|$ is the size of input $x$, $p$ is a parameter, and $f$ is an arbitrary computable function.
Hence, if $p$ is fixed and $f(p)$ is relatively small, the problem can still be considered as tractable.
The $W$ hierarchy is a collection of complexity classes of parameterized problems to capture the fixed-parameter intractability. For $i\geq 0$, $W[i]\subseteq W[i+1]$ and $W[0]$ is the set of fixed-parameter tractable problems. 
$W[1]$ can be regarded as a parameterized version of the complexity class NP.
Every $W[j]\,(j\geq 1)$-hard problem is believed to be fixed-parameter intractable.
For more details of the parameterized complexity we refer the reader to \cite{flum2006parameter,downey2012para}.

\begin{Proposition}\label{prop-fixed-para}
	The QCM problem parameterized by the number of qubits in $\Phi$ is $W[1]$-hard.
\end{Proposition}
\begin{proof}
	We make a fixed-parameter reduction from the clique problem to the QCM problem.
	The clique problem is that given an undirected graph $G$ and a number $n$, decide whether there is a clique of size $n$ in $G$. The clique problem is NP-complete, and it is $W[1]$-complete when parameterized by the size $n$ of the clique~\cite{downey2012para}.
	
	Given an instance $(G,n)$ of the clique problem, we reduce it to the instance $(\Phi^n_{clique},G,0)$ of the QCM problem. 
	$\Phi^n_{clique}$ is a quantum circuit with $n$ qubits and every pair of qubits are operated by a CNOT gate. It is easily seen that the topology graph of $\Phi^n_{clique}$ is a clique of size $n$. By the construction of $(\Phi^n_{clique},G,0)$ we see that the clique problem is fixed-parameter reducible to the QCM problem with the size of the clique and the number of qubits as parameters, respectively.
	It is obvious that $G$ contains a clique of size $n$ iff $\Phi^n_{clique}$ satisfies $G$'s NN constraint without using any SWAP gates.
\end{proof}

The QCM problem is an NP optimization problem aiming to find the least number of swap operations.
Since an NP optimization problem has a fully polynomial-time approximation scheme only if it is fixed-parameter tractable~\cite{downey2012para}.
Proposition~\ref{prop-fixed-para} implies the hardness of the QCM problem.

\subsection{The complexity of QCM on constrained quantum circuits and coupling graphs}\label{subsec-depthcoupl}

Shallow quantum circuits are the kind of quantum circuits with fixed depth. They are strictly more expressive than the classical shallow circuits~\cite{bravyi2020}. Due to their robustness to noise and decoherence, shallow quantum circuits are easily to implement on NISQ devices. In the following we consider the QCM problem on shallow quantum circuits.
First we prove a proposition about the degree of shallow quantum circuits' topology graphs.

\begin{Proposition}\label{prop-shallow-degree}
	The topology graph of every quantum circuit with maximum depth $d$ is a graph with maximum degree $d$, and every graph with maximum degree $d$ is the topology graph of a quantum circuit with maximum depth $d+1$.
\end{Proposition}
\begin{proof}
	Suppose that $\Phi$ is a quantum circuit with maximum depth $d$. Since a qubit is operated by at most one two-qubit gate in a layer, every qubit in $\Phi$ is operated by at most $d$ two-qubit gates. Hence, the topology graph of $\Phi$ is of maximum degree $d$ . 
	
	Conversely, given a graph $H$ with maximum degree $d$, we construct a quantum circuit $\Phi$ with maximum depth $d+1$ and its topology graph is $H$. By Vizing's theorem, we can use at most $d+1$ colors to color every edge in $H$ such that no two incident edges have the same color. So the edges in $H$ can be partitioned into at most $d+1$ sets where the edges in the same set do not have common endpoints. Define the numbers of qubits and layers of $\Phi$ to equal the numbers of nodes and colored edge sets of $H$, respectively, and two qubits are operated by a CNOT gate in the $i$-th layer iff there is an edge connecting them in the $i$-th edge set. It follows easily that $H$ is the topology graph of $\Phi$.
\end{proof}

\begin{Proposition}\label{prop-qcm-np-shallow}
	The QCM problem is NP-complete on shallow quantum circuits and bounded degree coupling graphs. More precisely, it is NP-complete if
	\begin{enumerate}[label=(\arabic*), labelsep=0.2cm]
		\item the quantum circuit has maximum depth 3, and the coupling graph is a planar bipartite and 3-degree bounded graph, or
		\item the quantum circuit has maximum depth 2, and the coupling graph is a 4-degree bounded graph.
	\end{enumerate}
\end{Proposition}

\begin{proof}
	We make reductions from the Hamiltonian cycle problem which is NP-complete for planar bipartite graphs with maximum degree 3~\cite{itai1982}.
	The undirected path and cycle graphs have maximum degree 2. By Proposition~\ref{prop-shallow-degree}, there are quantum circuits with maximum depth 3 such that their topology graphs are path or cycle graphs.
	Let $n\geq 1$, define quantum circuit $\Phi^n_{path} = \langle Q, \varGamma \rangle$, where
	\begin{enumerate}[labelsep=0.2cm]
		\item $Q = \{q_1,\dots,q_n\}$, and $\Gamma = (g_1,\dots,g_{n-1})$ is a sequence of CNOT gates,
		\item for every $q_i\in Q\,(1\leq i < n)$, the qubits $q_i,q_{i+1}$ are operated by the gate $g_i$,
		\item the gates $\{g_i\mid i\text{ is odd}\}$ are in the 1st layer and the other gates are in the 2nd layer.
	\end{enumerate}

	Define quantum circuit $\Phi^n_{cycle}$ by adding a CNOT gate $g_n$ to $\Phi^n_{path}$ such that $g_n$ is the last gate to execute in the circuit and operates on $q_1,q_n$. Figure~\ref{fig-reduct} shows the quantum circuits $\Phi^n_{path}$ and $\Phi^n_{cycle}$ where $n$ is odd.  It is easily seen that the topology graphs of $\Phi^n_{path}$ and $\Phi^n_{cycle}$ are a path and a cycle, respectively. 
	We have $depth(\Phi^n_{path})=2$, and $depth(\Phi^n_{cycle})=2$ if $n$ is even and $depth(\Phi^n_{cycle})=3$ if $n$ is odd.
	
	\begin{figure}[h]
		\centering
		\begin{adjustbox}{width=0.25\textwidth}
			\begin{quantikz}
				\lstick{$q_1$} & \ctrl{1}  & \qw      & \qw \\
				\lstick{$q_2$} & \targ{}   & \ctrl{1} & \qw \\
				\lstick{$q_3$} & \ctrl{1}  & \targ{}  & \qw \\
				\lstick{$q_4$} & \targ{}   & \qw      & \qw \\
				&\vdots&&\\
				\lstick{$q_{n-2}$} & \ctrl{1} & \qw      & \qw \\
				\lstick{$q_{n-1}$} & \targ{}  & \ctrl{1} & \qw \\
				\lstick{$q_{n}$}   & \qw      & \targ{}  & \qw 
			\end{quantikz}
			\hspace{1cm}
			\begin{tikzcd}[every arrow/.append style={dash}]
				\tikz{\node[draw, circle, inner sep=2pt]{$q_1$}} \arrow[d]  \\
				\tikz{\node[draw, circle, inner sep=2pt]{$q_2$}} \arrow[d]  \\
				\tikz{\node[draw, circle, inner sep=2pt]{$q_3$}} \arrow[d]  \\
				\raisebox{0.1cm}{\vdots} \\
				\tikz{\node[draw, circle, inner sep=2pt]{$q_n$}} \arrow[u]  
			\end{tikzcd}
		\end{adjustbox}
		\hspace{1cm}
		\begin{adjustbox}{width=0.35\textwidth}
			\begin{quantikz}
				\lstick{$q_1$} & \ctrl{1}  & \qw      & \ctrl{7} 	 & \qw \\
				\lstick{$q_2$} & \targ{}   & \ctrl{1} & \qw  		 & \qw\\
				\lstick{$q_3$} & \ctrl{1}  & \targ{}  & \qw  		 & \qw\\
				\lstick{$q_4$} & \targ{}   & \qw      & \qw  		 & \qw\\
				&\vdots&&\\
				\lstick{$q_{n-2}$} & \ctrl{1} & \qw      & \qw 	 & \qw\\
				\lstick{$q_{n-1}$} & \targ{}  & \ctrl{1} & \qw  	 & \qw\\
				\lstick{$q_{n}$}   & \qw      & \targ{}  & \targ{}  & \qw
			\end{quantikz}
			\hspace{0.8cm}
			\begin{tikzpicture}[baseline=0cm]
				\def \radius {1.2cm}
				\def \margin {14} 
				\foreach \s in {1,...,3}
				{
					\node[draw, circle, inner sep=2pt] at ({72 * (\s - 1)}:\radius) {$q_\s$};
					\draw[black] ({72 * (\s - 1)+\margin}:\radius) arc ({72 * (\s - 1)+\margin}:{72 * (\s)-\margin}:\radius);
				}
				\node[draw, circle, inner sep=2pt] at (216:\radius) {$q_4$};
				\draw[black,dashed] ({216+\margin}:\radius) arc ({216+\margin}:{288-\margin}:\radius);
				\node[draw, circle, inner sep=2pt] at (288:\radius) {$q_n$};
				\draw[black] ({288+\margin}:\radius) arc ({288+\margin}:{360-\margin}:\radius);
			\end{tikzpicture}
		\end{adjustbox}
		\subfloat[\label{fig-reduct-a}]{\hspace{0.4\linewidth}}
		\subfloat[\label{fig-reduct-b}]{\hspace{0.4\linewidth}}
		\medskip
		\caption{\label{fig-reduct} An illustration of (\textbf{a}) the quantum circuit $\Phi^n_{path}$ whose topology graph is a path, and (\textbf{b}) the quantum circuit $\Phi^n_{cycle}$ whose topology graph is a cycle.}
	\end{figure}
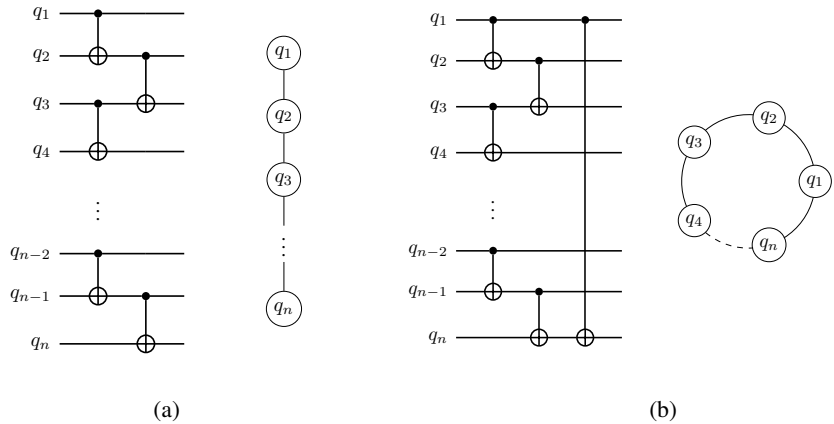

	To prove (1),  we reduce an instance $G$ of the Hamiltonian cycle problem, where $G$ is a planar bipartite and 3-degree bounded graph, to the instance $(\Phi^n_{cycle},G,0)$ of the QCM problem, where $n$ is the number of nodes in $G$. $\Phi^n_{cycle}$ has maximum depth 3. $G$ has a Hamiltonian cycle iff $\Phi^n_{cycle}$ satisfies $G$'s NN constraint without using any SWAP gates.
	
	To prove (2), given a graph $G$ with maximum degree 3 and $n$ nodes, we make the reduction as follows.
	\begin{enumerate}[labelsep=0.2cm]
		\item If there is a node $a$ with degree 2 in $G$, then define a graph $H$ by attaching two new nodes to $a$ and one of its adjacent nodes, respectively. The two new nodes have degree 1, and the degree of $H$ is at most 4, since the degree of one neighbor of $a$ is increased exactly by 1. We reduce $G$ to the instance $(\Phi^{n+2}_{path},H,0)$ of QCM.
		\item If no node in $G$ has degree 2,  then $H$ is obtained by attaching a new node to a node $a$ in $G$, and connecting one endpoint $b$ of a new edge $(b,c)$ to the neighbors of $a$. The degree of $b$ is at most 4, the degrees of $c$ and the new node connected to $a$ are 1, and the degrees of neighbors of $a$ are increased by 1. So $H$ has maximum degree 4. We reduce $G$ to the instance $(\Phi^{n+3}_{path},H,0)$ of QCM.
	\end{enumerate}
	$G$ has a Hamiltonian cycle iff $H$ has a Hamiltonian path (the two new nodes with degree 1 are the start and end points, respectively) iff $\Phi^{n+2}_{path}$ (or $\Phi^{n+3}_{path}$) satisfies $H$'s NN constraint without using any SWAP gates.
\end{proof}

The newest IBM quantum devices have adopted the heavy-hex lattice architecture, in which each unit cell consists of a hexagonal arrangement of qubits. The heavy-hex lattice is equivalent to the square lattice up to a constant overhead~\cite{ibmhex2021}.
Both of them are grids constituted of regular unit graphs.
The squares, regular hexagons and equilateral triangles are the only three kinds of regular polygons that can tile the plane by themselves.

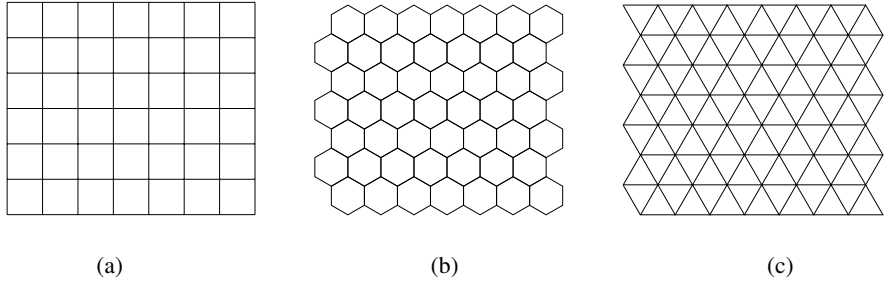
\begin{figure}[h]
	\centering
	\begin{adjustbox}{width=0.2\textwidth}
		\begin{tikzpicture}
		\draw (0,0) grid (7,6);
		\end{tikzpicture}
	\end{adjustbox}
	\qquad
	\begin{adjustbox}{width=0.2\textwidth}
		\begin{tikzpicture} [hexa/.style= {shape=regular polygon,regular polygon sides=6, minimum size=1cm,draw,rotate=30}]
		\foreach \j in {0,...,6}{
			\foreach \i in {0,...,6}{\node[hexa] at ({(\i-Mod(\j,2)/2)*sin(60)},{\j*0.75}){};}  
		}   
		\end{tikzpicture}
	\end{adjustbox}
	\qquad
	\begin{adjustbox}{width=0.21\textwidth}
		\begin{tikzpicture} [triag/.style= {shape=regular polygon,regular polygon sides=3, minimum size=1cm,draw,rotate=0}]
		\foreach \j in {0,...,6}{
		\foreach \i in {0,...,6}{\node[triag] at ({(\i-Mod(\j,2)/2)*sin(60)},{\j*0.75}){};}  
		}   
		\draw (-0.88,5) -- ++(0:6.08);
		\draw (-0.44,-0.25) -- ++(120:0.85);
		\draw (-0.44,1.25) -- ++(120:0.85);
		\draw (-0.44,2.75) -- ++(120:0.85);
		\draw (-0.44,4.25) -- ++(120:0.85);
		\draw (5.63,1.25) -- ++(60:-0.85);
		\draw (5.63,2.75) -- ++(60:-0.85);
		\draw (5.63,4.25) -- ++(60:-0.85);
		\end{tikzpicture}
	\end{adjustbox}
	\subfloat[\label{fig-regular-a}]{\hspace{0.27\linewidth}}
	\subfloat[\label{fig-regular-b}]{\hspace{0.27\linewidth}}
	\subfloat[\label{fig-regular-c}]{\hspace{0.27\linewidth}}
	\medskip
	\caption{\label{fig-regular} There kinds of grids that are made of (\textbf{a}) the squares, (\textbf{b}) the regular hexagons and (\textbf{c}) the equilateral triangles.}
\end{figure}

The grid graph is a finite induced subgraph of the (infinite) grid shown in Figure~\ref{fig-regular}. 
Using a similar proof of Proposition~\ref{prop-qcm-np-shallow}, the following proposition can be proved by the fact that the Hamiltonian cycle problem is still NP-complete on square grid graphs~\cite{itai1982}, regular hexagon grid graphs~\cite{islam2007Hamcyc} and equilateral triangle graphs~\cite{valentin2006hamtriag,valery2008hamtriag}.

\begin{Proposition}
	The QCM problem is NP-complete on grid coupling graphs.
\end{Proposition}

\subsection{The complexity of QCM with fixed number of SWAP operations}\label{subsec-numberofswap}

The number of swap operations is another parameter of the QCM problem. By Lemma~\ref{lem-max-swap} we know that for any quantum circuit $\Phi$ and coupling graph $G$, an upper bound of it can be computed easily. We prove that the complexity of QCM with a fixed number of swap operations is still NP-complete.

\begin{Proposition}
	The QCM problem is NP-complete for each fixed $k\geq 0$.
\end{Proposition}
\begin{proof}
	It is trivial for the case $k=0$, where the subgraph isomorphism problem can be reduced to this problem. 
	
	Suppose that $k\geq 1$. We make a reduction from the Hamiltonian cycle problem that is still NP-complete on graphs of maximum degree 3~\cite{garey1976}. The proof is similar to that of Proposition~\ref{prop-nl-comple}.
	Given an undirected graph $H$ of maximum degree 3, we construct an instance $(\Phi,G,k)$ such that there is a Hamiltonian cycle in $H$ iff $\Phi$ can be transformed to satisfy $G$'s NN constraint using at most $k$ SWAP gates. 
	Define $\Phi$ to contain parallelly two quantum circuits $(\Phi^5_{clique})^{*2k}$ and $(\Phi^n_{cycle})^{*2k}$ followed by a CNOT gate that operates on the first qubits of them, where $n$ is the number of nodes in $H$. Define $G$ to be a graph that connects a clique of size 5 to a node in $H$ using a path of length $k+1$. 
	
	If there is a Hamiltonian cycle in $H$, we can map the qubits of $(\Phi^5_{clique})^{*2k}$ to the clique of size 5 in $G$, and the qubits of $(\Phi^n_{cycle})^{*2k}$ to the Hamiltonian cycle in $H$. To make the last CNOT gate nearest neighbor, we can move the qubits operated by it along the path between the clique and the cycle, where at most $k$ SWAP gates are enough.
	If there is not a Hamiltonian cycle in $H$, it needs more than $k$ SWAP gates whatever we transform $\Phi$ on $G$ by Lemma~\ref{lem-least-swap}. So if $\Phi$ can be transformed to satisfy $G$'s NN constraint using at most $k$ SWAP gates, then there must be a Hamiltonian cycle in $H$.
\end{proof}

\section{Discussion}\label{sec-discus}
Quantum circuit mapping is an important procedure for running quantum circuits on NISQ devices. It transforms quantum circuits to be compliant with the nearest neighbor constraint by adding SWAP gates. The QCM problem is an NP-complete optimization problem aiming to find the minimal number of SWAP gates. So it is unlikely to get a polynomial time algorithm for it.
We study the parameterized complexity of QCM in the paper. First we give an exact algorithm that computes the minimal number of SWAP gates. The complexity analysis shows that the algorithm runs in polynomial time if the coupling graph is fixed. And if the number of qubits of the quantum circuit is fixed to a constant, the QCM problem is NL-complete, which is believed strictly below the complexity class P. We prove by a reduction from the undirected shortest path problem for graphs with maximum degree 3. Further, taking the number of qubits of the quantum circuit as a parameter, we show that QCM is W[1]-hard by a reduction from the clique problem. Every problem in W[1] is considered as fixed-parameter intractable. Hence, the QCM problem is unlikely to have a fully polynomial-time approximation scheme.

The depth of the quantum circuits and the type of the coupling graphs are the other two parameters considered in the paper. We show that the QCM problem is still NP-complete over shallow quantum circuits, and planar, bipartite and degree bounded coupling graphs. These results indicate that the number of qubits is the key factors that affect the complexity of QCM. Efficient algorithms for the quantum circuit mapping from quantum circuits with a reasonable number of qubits to a fixed quantum device are theoretically possible, and actually many such algorithms already exist. Algorithm~\ref{algo-compute-fun} can be easily adapted to save the swap actions to get a solution while computing the minimal number. But it is not practical since the result in Figure~\ref{fig-experiment} shows that the time increases significantly as the number of qubits becomes larger. So finding a novel algorithm is one of the future work.

All active IBM Quantum devices have adopted the heavy-hex lattice architecture. It is a kind of grid constituted of regular hexagons with high scalability. We also prove that QCM is NP-complete on three kinds of grid coupling graphs that are finite induced subgraphs of the (infinite) grids made up of the squares, regular hexagons and equilateral triangles, respectively. But it is still open for the complexity of QCM over some simple coupling graphs, e.g., the linear and cycle coupling graphs. Furthermore, if the coupling graph is fixed, the complexity of QCM is obviously in NL by Proposition~\ref{prop-nl-comple}, we conjecture that it is in LOGSPACE. This means that the effective quantum circuit mapping for a specified quantum device is always possible. Whether QCM is fixed-parameter tractable parameterized by the number of nodes in the coupling graph is another open problem. These can be the future work of the paper.

\bibliographystyle{plain}
\bibliography{./quantumref, ./reference}
\end{document}